\documentclass[11pt]{article}
\usepackage{lipsum}

\newcommand\blfootnotea[1]{%
  \begingroup
  \renewcommand\thefootnote{}\footnote{#1}%
  \endgroup
}

\usepackage[utf8]{inputenc} %
\usepackage[T1]{fontenc}    %
\usepackage{url}            %
\usepackage{booktabs}       %
\usepackage{amsfonts}       %
\usepackage{nicefrac}       %
\usepackage{microtype}      %
\usepackage{xcolor}         %

\usepackage{amsthm,amsfonts,amsmath,amssymb,epsfig,color,float,graphicx,verbatim,dsfont,algorithm}

\usepackage[inline]{enumitem}
\usepackage{thm-restate}

\usepackage{mathtools}
\usepackage{bbm}
\usepackage[colorlinks,citecolor=blue,linkcolor=magenta,bookmarks=true]{hyperref}
\usepackage[nameinlink]{cleveref}

\crefname{equation}{equation}{equations}

\crefname{lemma}{lemma}{lemmata}
\crefname{claim}{claim}{claims}
\crefname{theorem}{theorem}{theorems}
\crefname{proposition}{proposition}{propositions}
\crefname{corollary}{corollary}{corollaries}
\crefname{claim}{claim}{claims}
\crefname{remark}{remark}{remarks}
\crefname{definition}{definition}{definitions}
\crefname{fact}{fact}{facts}
\crefname{question}{question}{questions}
\crefname{condition}{condition}{conditions}
\crefname{algorithm}{algorithm}{algorithms}

\newtheorem{theorem}{Theorem}[section]
\newtheorem{lemma}[theorem]{Lemma}

\newtheorem{proposition}[theorem]{Proposition}

\newtheorem{definition}[theorem]{Definition}

\newtheorem{fact}[theorem]{Fact}

\newtheorem{problem}{Problem}

\newcommand{\eps}{\epsilon}

\newcommand{\poly}{\mathrm{poly}}

\newcommand{\I}{{\mathbb I}}

\newcommand{\trace}{\operatorname{tr}}

\newcommand{\Var}{\operatorname{Var}}

\def\P{\mathbb P}
\def\R{\mathbb R}

\def\N{\mathbb N}
\def\I{\mathbb I}

\newcommand{\cA}{\mathcal{A}}

\newcommand{\cD}{\mathcal{D}}
\newcommand{\cE}{\mathcal{E}}
\newcommand{\cF}{\mathcal{F}}

\newcommand{\cP}{\mathcal{P}}

\newcommand{\cU}{\mathcal{U}}

\newcommand{\cX}{\mathcal{X}}

\newcommand{\Chi}{\cX}

\DeclareMathOperator*{\E}{\mathbb{E}}

\renewcommand{\Pr}{\operatorname*{\mathbb{P}}}

\newcommand{\Exp}{\operatorname*{\mathbb{E}}}

\newcommand{\muhat}{\hat{\mu}}

\newcommand{\Normal}{\mathcal{N}}

\newcommand{\Bernoulli}{\mathrm{Ber}}

\newcommand{\DTV}{d_\mathrm{TV}}

\newcommand{\DK}{d_\mathrm{K}}
\renewcommand{\boxed}[1]{\text{\fboxsep=.2em\fbox{\m@th$\displaystyle#1$}}}

\newcommand{\res}{B}
\newcommand{\AKPprob}{q}
\newcommand{\AKPtail}{s_1}
\newcommand{\Qconst}{s_2}
\newcommand{\Qscale}{s_3}
\newcommand{\InftyBd}{r}

\newcommand{\alphacoordinates}{3\alpha}
\newcommand{\alphacoordinateslp}{\alpha}

\newcommand{\Real}{\mathbb{R}}

\newcommand{\1}{\mathds{1}}

\renewcommand{\d}{\mathrm{d}}

\newcommand{\tr}{\mathrm{tr}}

\def\d{\mathrm{d}}

\def\colorful{0}

\ifnum\colorful=1
\newcommand{\new}[1]{{\color{red} #1}}

\else
\newcommand{\new}[1]{{#1}}

\fi

\title{Outlier-Robust Sparse Mean Estimation for Heavy-Tailed Distributions\blfootnotea{Authors are in alphabetical order. Part of this work was done while a subset of the authors were visiting the Simons Institute for the Theory of Computing.}}

\author{
Ilias Diakonikolas\thanks{Supported by NSF Medium Award CCF-2107079, NSF Award CCF-1652862 (CAREER), a Sloan Research Fellowship, and a DARPA Learning with Less Labels (LwLL) grant.}\\
University of Wisconsin-Madison\\
{\tt ilias@cs.wisc.edu}\\
\and
Daniel M. Kane\thanks{Supported by NSF Medium Award CCF-2107547, NSF Award CCF-1553288 (CAREER), and a grant from CasperLabs.}\\
University of California, San Diego\\
{\tt dakane@cs.ucsd.edu}
\and
Jasper C.H. Lee\thanks{Supported in part by the generous funding of a Croucher Fellowship for Postdoctoral Research, NSF award DMS-2023239, NSF Medium Award CCF-2107079 and NSF AiTF Award CCF-2006206.}\\
University of Wisconsin-Madison\\
{\tt jasper.lee@wisc.edu}\\
\and
Ankit Pensia\thanks{Supported by NSF grants NSF Award CCF-1652862 (CAREER), DMS-1749857, and CCF-1841190.}\\
University of Wisconsin-Madison\\
{\tt ankitp@cs.wisc.edu}\\
}

\begin{document}

\maketitle

\setcounter{page}{0}
\thispagestyle{empty}

\allowdisplaybreaks

\begin{abstract}
We study the fundamental task of outlier-robust mean estimation 
for heavy-tailed distributions in the presence of sparsity. 
Specifically, given a small number of corrupted samples from a high-dimensional heavy-tailed distribution 
whose mean $\mu$ is guaranteed to be sparse, 
the goal is to efficiently compute a hypothesis that accurately approximates $\mu$ with high probability.
Prior work had obtained efficient algorithms for robust sparse mean estimation of light-tailed distributions.
In this work, we give the first sample-efficient and polynomial-time robust sparse mean estimator 
for heavy-tailed distributions under mild moment assumptions. 
Our algorithm achieves the optimal asymptotic error 
using a number of samples scaling logarithmically with the ambient dimension.
Importantly, the sample complexity of our method is optimal 
as a function of the failure probability $\tau$, having an {\em additive} $\log(1/\tau)$ dependence. 
Our algorithm leverages the stability-based approach from the algorithmic robust statistics literature, 
with crucial (and necessary) adaptations required in our setting.
Our analysis may be of independent interest, 
involving the delicate design of a (non-spectral) decomposition 
for positive semi-definite matrices satisfying certain sparsity properties.
\end{abstract}

\newpage

\section{Introduction}

\subsection{Background}

One of the most fundamental problem setups in statistics is as follows: 
given $n$ i.i.d.~samples drawn from an unknown distribution $P$ 
chosen arbitrarily from some known distribution family $\cP$, 
infer some particular property of $P$ from the data.
This generic model captures a range of statistical problems of interest, 
for example, parameter estimation (such as the mean and (co)variance of $P$), 
as well as hypothesis testing.
While long lines of work have given us a deep understanding of the statistical 
and computational possibilities and limitations on these problems, 
these results are not always applicable in real-world settings due to 
(i) modeling issues, that the underlying distribution $P$ might not actually 
be in the known family $\cP$ but only being close to it, and 
(ii) the fact that the $n$ samples supplied might be corrupted, 
for example by nefarious actors in high-stakes applications~\cite{AndBHHRT72}.

The field of \emph{robust statistics} aims to design estimators and testers 
that can tolerate up to a \emph{constant} fraction of corrupted samples, 
independent of the potentially high dimensionality of the data~\cite{Tukey60,Huber09}.
Classical works in the field have identified and resolved the statistical limits 
of problems in this setup, both in terms of constructing estimators and proving impossibility results~\cite{Yatracos85,DonLiu88a,Donoho92,Huber09}. However, the proposed estimators 
were not computationally efficient, often requiring exponential time 
to compute either in the number of samples or the number of dimensions~\cite{Huber09}.

A recent line of work, originating in the computer science community, 
has developed the subfield of \emph{algorithmic} robust statistics, 
aiming to design estimators that not only attain tight statistical guarantees, 
but are also computable in polynomial time.
This line of research has provided computationally and statistically efficient estimators 
in a variety of problem settings (e.g., mean estimation, covariance estimation, and linear regression) 
under different assumptions (e.g., the distribution might be assumed 
to be (sub-)Gaussian, or can be heavy tailed); 
see \cite{DK19-survey} for a recent survey of results.

The focus of this paper is the robust mean estimation problem 
under sparsity constraints on the mean vector.
Sparsity is an important structural constraint that is both relevant in practice, 
especially in the face of increasing dimensionality of modern data, 
and extensively studied for statistical estimation 
(see, e.g., the books \cite{Hastie15,eldar2012compressed,vandeGeer16}).
In the specific context of robust {sparse} mean estimation, 
prior works have studied the case where the underlying distribution has light-tails, 
e.g., sub-exponential tails~\cite{BDLS17,DKKPS19-sparse,CDKGS21-sparse,DKKPP22-sparse-sos}.
In particular, the case of a spherical Gaussian distribution 
is now rather well-understood both in terms of the optimal information-theoretic estimation error, 
as well as the conjectured \emph{computational-statistical tradeoff} --- namely, 
that there is a gap between the statistical performance 
of computationally efficient and inefficient estimators~\cite{DKS17-sq,BrennanB20}.
In this work, we initiate the investigation of outlier-robust sparse mean estimation 
for \emph{heavy-tailed} distributions, under only mild moment assumptions.
Our main result is the first computationally efficient robust mean estimator in the heavy-tailed setting 
which leverages sparsity to reduce the {sample complexity} from depending polynomially 
on the dimensionality to a logarithmic dependence.
Importantly, our algorithm also achieves the optimal dependence 
on the failure probability $\tau$ as it tends to $0$; 
see the next two subsections for further discussion.

\subsection{Problem Setup}
We first define the input contamination model before formally stating the statistical problem.

\begin{definition}[Strong Contamination Model]
Given a \emph{corruption} parameter $\eps \in (0,1/2)$ 
and a distribution $P$ on uncorrupted samples, 
an algorithm takes samples from $P$ with \emph{$\eps$-contamination} 
as follows: 
(i) The algorithm specifies the number $n$ of samples it requires. 
(ii) $n$ i.i.d.\ samples from $P$ are drawn but not yet shown to the algorithm. 
(iii) An arbitrarily powerful adversary then inspects 
the entirety of the $n$ i.i.d.\ samples, 
before deciding to replace any subset of $\lceil \eps n \rceil$ 
samples with arbitrarily corrupted points, 
and returning the modified set of $n$ samples to the algorithm.
\end{definition}

Define the $\ell_{2,k}$-norm of a vector $v$, denoted by $\|v\|_{2,k}$, 
as the $\ell_2$-norm of the largest $k$ entries of 
a vector $v$ in magnitude.
The goal is to estimate the mean vector in this sparse norm.

\begin{problem}
\label{prob:1}
Fix a corruption parameter $\eps\in(0,1/2)$, error parameter $\delta>0$, 
failure probability $\tau\in(0,1)$, and distribution family $\cD$ over $\Real^d$.
Suppose we have access to $\eps$-contaminated samples 
drawn from an unknown distribution $P \in \cD$ with mean $\mu$.
The task is to compute an estimate $\muhat$ such that $\|\muhat - \mu\|_{2,k}$ 
is bounded above by error $\delta$ with probability at least $1-\tau$ over $n$ samples.
The goal is then to give an estimator with the minimal \emph{sample complexity} $n(k,\eps,\delta,\tau)$.
\end{problem}

The above problem is slightly more general than sparse mean estimation in the following sense.
To estimate a $k$-sparse mean vector $\mu$ to error $\delta$, 
it suffices (see, e.g.,~\cite[Lemma 3.2]{CDKGS21-sparse}) to:
1) obtain an estimate $\tilde{\mu}$ with $\|\tilde{\mu}-\mu\|_{2,k} \le \delta/3$, 
and 2) round $\tilde{\mu}$ to the $k$ entries with the largest magnitude, and zero out all the other entries.
The main result of this paper solves the problem of robust mean estimation in the $\ell_{2,k}$ norm.

A key aspect of robust statistics is that, depending on the distribution family $\cD$ we consider, 
the above problem is generally not solvable for all error parameters $\delta > 0$.
This work focuses on sparse mean estimation for \emph{heavy-tailed} distributions, 
where a commonly used model for heavy-tailedness is imposing only the mild assumption 
that the covariance of the clean distribution is bounded by the identity $I$, 
without any further tail assumptions (see \Cref{sec:literature} for more discussion).
Even when $d=1$ and even when there are infinitely many samples~\cite{DK19-survey}, 
it is known that in the heavy-tailed setting the minimum $\delta$ achievable is in the order of $\sqrt{\eps}$.
This immediately implies the same lower bound of $\Omega(\sqrt{\eps})$ for the minimum achievable $\delta$ in \Cref{prob:1}.

Before discussing the algorithmic results in this paper, 
we state known information-theoretic bounds on the sample complexity 
that applies to all estimators, efficient or not, for \Cref{prob:1} 
on distributions with covariance bounded by $I$, 
and for $\delta = \Theta(\sqrt{\eps})$.

\begin{fact}[Information-theoretic sample complexity: computationally-inefficient] \label{fact:info-theoretic}
In \Cref{prob:1}, for the distribution family $\cD_2$ defined as the set of distributions 
with covariance $\Sigma \preceq I$, and for $\delta = \Theta(\sqrt{\eps})$, 
we have that $n(k,\eps,\delta,\tau)  \asymp  (k \log(d/k) + \log(1/\tau))/\eps$.
That is, any estimator requires at least these 
many samples and there exists a (computationally-inefficient) 
estimator with this sample complexity.
The upper bound is from~\cite{Dep20-vc,PrasadBR20} 
and the lower bound follows from~\cite{LugMen19-gen}, 
even in the absence of outliers (see also Footnote 2 in~\cite{DepLec22-pseudo}) 
and even when we restrict to the distribution family $\cD_{\textrm{Gaussian}}$ 
which is the set of the Gaussian distributions with identity covariance.
\end{fact}

An interesting aspect of robust sparse mean estimation 
is that there is a conjectured statistical-computational tradeoff, 
namely that efficient algorithms require a qualitatively larger sample complexity than inefficient ones.
Specifically, there is evidence (in the form of SQ lower bounds and reduction-based hardness) that 
all efficient algorithms have a quadratically worse dependence on $k$; 
that is, even for constant $\eps, \delta, \tau$, and $\cD_{\textrm{Gaussian}}$ 
being identity-covariance Gaussians in \Cref{prob:1}, 
the sample complexity of all efficient algorithms is at least $\tilde{\Omega}(k^2)$, 
as opposed to $\tilde{O}(k)$ in \Cref{fact:info-theoretic}. 
See~\cite{DKS17-sq,BrennanB20} for a detailed discussion.

Both the information-theoretic bound and the conjectured information-computation tradeoff 
serve as benchmarks for our algorithm to match.

The main result of this paper is the following.

\begin{theorem}[Computationally Efficient Heavy-Tailed Robust Sparse Mean Estimation] \label{thm:ourResult}
 Let $\epsilon \in (0, \epsilon_0)$ for some sufficiently small universal constant $\epsilon_0>0$. 
 Let $P$ be a distribution over $\R^d$, where the mean and covariance of $P$ are $\mu$ and $\Sigma$ respectively.
 Suppose $\Sigma \preceq I$ and further suppose that 
 for all $j \in [d]$, $\E[(X_j-\mu_j)^4] = O(1)$.
  Then there is an algorithm such that on input 
  (i) the corruption parameter $\epsilon$, 
  (ii) the failure probability $\tau$, 
  (iii) the sparsity parameter $k$, and 
  (iv) $T$, an $\eps$-corrupted set of  
   $n \gg  (k^2\log d + \log(1/\tau))/\eps)$ i.i.d.~samples from $P$, 
   the algorithm outputs $\widehat{\mu}$ satisfying 
   $\|\widehat{\mu} -  \mu\|_{2,k}$ $= O(\sqrt{\eps})$ 
   with probability $1 - \tau$ in $\poly(n,d)$ time.
\end{theorem}

Phrased in a slightly different language, when our estimator is given 
a sufficiently large number $n$ of $\eps$-corrupted samples, 
it outputs an estimate $\muhat$ satisfying 
$\|\widehat{\mu} -  \mu\|_{2,k} = O\Big(\sqrt{\frac{k^2 \log d}{n}} + \sqrt{\epsilon} + \sqrt{\frac{\log(1/\tau)}{n}}\Big)$ 
with probability $1-\tau$.

We note that the guarantees of our algorithm remain the same 
under a weaker assumption on $\Sigma$: we need only that 
$\|\Sigma\|_{\Chi_k} \leq 1$ instead of the spectral norm being bounded
(the norm $\|\cdot\|_{\Chi_k}$ is formally defined in \Cref{def:chi-k-norm}).
Informally, the $\Chi_k$ norm of a square matrix $A$ is a convex relaxation 
of finding the maximum of $v^\top A v$ over $k$-sparse vectors $v$.
See \Cref{thm:ourResult_Chi_k} in \Cref{sec:alg_analysis} for the stronger 
version of the main result, which assumes only that $\|\Sigma\|_{\Chi_k} \le 1$.

As outlined above, the dependence of our sample complexity result on $k$ 
is tight with respect to the conjectured lower bound for efficient algorithms, 
and its dependence on $\tau$ and $\eps$ are also tight with respect 
to the information-theoretic lower bounds, even in the Gaussian case.
In terms of the smallest achievable asymptotic error (even in infinite sample regime), 
we show in \Cref{lem:InfoTheoretic4thMoment} that, 
even after adding the mild axis-wise $4^\text{th}$ moment assumption in \Cref{thm:ourResult}, 
the asymptotic error remains bounded below by $\Omega(\sqrt{\eps})$ when $k$ is sufficiently large.
The restriction on $k$ is fairly mild, covering most parameter regimes of interest.

The sample complexity of our algorithm has a dependence on the failure probability that is $\log 1/\tau$, 
and --- importantly --- this is an {\em additive} term in the complexity instead of multiplicative.
To be precise, in the i.i.d.\ setting with no outliers, 
we can artificially set $\epsilon = C \max(k^2 \log d, \log (1/\tau)) /n $ for a large constant $C$.
In this setting, when the number of samples $n$ is such that $n \gg k^2 \log d + \log (1/\tau)$, 
then with probability $1-\tau$, our algorithm outputs an estimate $\hat \mu$ satisfying
\begin{align}
    \| \hat \mu - \mu \|_{2,k} = O\left(\sqrt{\frac{k^2 \log d}{n}} + \sqrt{\frac{\log(1/\tau)}{n}} \right) \;.
\end{align}
This additive dependence is non-trivial to achieve 
even in the optimal rates for heavy-tailed mean estimation in the non-robust (and non-sparse) setting.
See the~\cite{LugMen19-survey} survey for a more detailed discussion.
Our work provides the first {\em computationally efficient} estimator 
for heavy-tailed sparse mean estimation 
with such additive dependence, even in the non-robust setting.

\subsection{Our Approach}
\label{sec:our-approach}
\looseness=-1
Our algorithm fits into the stability-based filtering approach;
see~\cite{DKKLMS16} and the survey \cite{DK19-survey}.
The filtering framework is a by-now-standard algorithmic 
technique in robust statistics. 
The approach for robust mean estimation can be summarized as follows: 
1) with high probability over the sampling of the $n$ uncorrupted samples, 
there exists a large subset of uncorrupted samples (say, a $1-O(\eps)$ fraction) 
satisfying a ``stability'' condition with respect to the mean of the uncorrupted distribution, 
and 2) a filtering algorithm taking as input an \emph{$\eps$-corrupted} version 
of the stable set of samples will remove some of the samples, 
such that the sample mean of the remaining points is guaranteed to be close to the true mean 
(which can then be returned as the final mean estimate). The notion of ``stability'' depends 
crucially on the task at hand, and is defined below for the sparse mean estimation problem. 

\paragraph{Stability-Based Algorithms under Sparsity}
Informally speaking, in the context of robust mean estimation, 
we say that a set $S$ is stable when the mean and the covariance of $S$ 
do not deviate too much when we remove a small fraction of elements from $S$.
For the task of {\em sparse} mean estimation, we would like to measure 
the deviation only along the $k$-sparse directions. 
However, it is computationally hard  to calculate the maximum 
of $v^\top A v$ over $k$-sparse unit vectors for an arbitrary matrix $A$ 
(this is known as the sparse PCA problem~\cite{TillPfe14}).
Following \cite{BDLS17}, our definition of stability involves a convex relaxation 
of the above optimization problem, using the following definition of the set $\Chi_k$ 
and the associated matrix norm $\|\cdot\|_{\Chi_k}$.

\begin{definition}[The set $\Chi_k$ and the norm $\|\cdot\|_{\Chi_k}$]
\label{def:chi-k-norm}
The set $\Chi_k$ is defined as the set of positive semidefinite matrices 
that have trace $1$ and $\ell_1$-norm at most $k$ when flattened as a vector.
The matrix norm $\|A\|_{\Chi_k}$ is then defined as 
{$\sup_{M \in \Chi_k} |  A \bullet M  |$, where $ A \bullet M$ denotes 
the trace product $\trace(A^\top M)$.}
\end{definition}

Note that for any square matrix $A$, $\|A\|_{\Chi_k}$ is always bounded above 
by its spectral norm. Furthermore, observe that for any square matrix $A$, 
the maximum of $v^\top Av$ over $k$-sparse unit vectors is bounded above 
by $\|A\|_{\Chi_k}$, and the latter can be calculated efficiently using a convex program.
We are now ready to define the stability condition for our sparse mean estimation task.

\begin{definition}[Stability Condition for Robust Sparse Mean Estimation] \label{def:stab-sparse}
For $0 < \eps < 1/2$ and $\eps \leq \delta$, a set $S$ is $(\epsilon,\delta,k)$-stable 
with respect to $\mu \in \R^d$ and $\sigma \in \R_+$ if it satisfies the following condition: 
for all subsets $S' \subset S$ with $|S'| \geq (1 -\epsilon)|S|$, the following holds:
	(i) 
	$\|\mu_{S'} - \mu\|_{2,k} \leq \sigma \delta$, 
and (ii)	 
$\|\overline{\Sigma}_{S'} - \sigma^2 I \|_{\cX_k} \leq \sigma^2 \delta^2/\epsilon$, 
where $\mu_{S'} = (1/|S'|)\sum_{x \in S'}x$ is the sample mean of $S'$ 
and $\overline{\Sigma}_{S'} = (1/|S'|)\sum_{x \in S'}(x - \mu)(x - \mu)^\top$ 
is the second moment of $S'$.
\end{definition}

\Cref{def:stab-sparse} is intended for distributions with covariance matrices 
at most $\sigma^2$ times the identity.
We will omit $\mu$ and $\sigma$ above when they are clear from the context.

Focusing on the class of identity covariance Gaussian distributions, \cite{BDLS17} 
gave a computationally-efficient algorithm for robust sparse mean estimation
using roughly $k^2\log d$ samples.\footnote{The additional factor of $k$ in their sample complexity (cf. \Cref{fact:info-theoretic}) is because the convex relaxation involving $\Chi_k$ norm can be loose. However, \cite{DKS17-sq,BrennanB20} suggest that $k^2$ samples are needed for efficient algorithms.} As we explain below, their algorithm succeeds 
under the stability condition of \Cref{def:stab-sparse}.

By using the standard median-of-means pre-processing 
described in \Cref{sec:prelim}, we can reduce the robust sparse mean estimation task 
to the case when the corruption parameter $\eps$ is constant, say $0.01$, 
and aim  to achieve only a constant estimator error in the $\ell_{2,k}$ norm.
For this regime, we state the guarantees of robust sparse mean estimation algorithm of \cite{BDLS17} (developed for the Gaussian setting) 
as follows\footnote{\new{See also \cite{ZJS22-generalized-gradients} for a related algorithm.}}:

\begin{fact} \label{fact:SDPAlg}
Let $S$ be a set in $\R^d$ such that there exists a set $S' \subseteq S$ 
such that (i) $|S'| \geq 0.99|S|$, and (ii) $S'$ is an $(0.01,O(1),k)$-stable 
with respect to (unknown) $\mu$ and (unknown) $\sigma$. 
There is a $\poly(|S|,d)$-time algorithm that takes as input $T$, 
an $0.01$-corruption of $S$, and returns a mean estimate $\muhat$ 
such that $\|\muhat - \mu\|_{2,k} \le O(\sigma)$.
\end{fact}

Given this prior algorithmic result, the key challenge is to show that, 
even in the setting of {\em heavy-tailed} data, a large subset 
of the uncorrupted samples satisfies the stability condition 
with high probability. Without sparsity constraints, \cite{DiaKP20,HopLZ20} 
showed that $O(d)$ samples suffice for stability 
(under a different definition appropriate for the dense setting), 
which is too large for our purposes.

\paragraph{Truncation is Necessary for Stability}
Recall that our goal is to show that if we draw roughly $k^2 \log d$ samples 
from a heavy-tailed distribution, then it contains a large stable subset. 
For the light-tailed data (Gaussian), this was shown in \cite{BDLS17}. 
However, this desired statement is not true for general heavy-tailed distributions.   
Consider the standard setting for modeling heavy-tailed data, namely 
that the covariance $\Sigma$ of the uncorrupted distribution 
is upper bounded by the identity.
For simplicity, also assume that the sparsity parameter $k$, 
corruption parameter $\eps$ and failure probability $\tau$ are all constants.
Thus, our goal is to show that, with high probability, there is a large stable 
subset among $\log d$ samples. 
Yet, as we show in \Cref{ex:HeavyTailed} in \Cref{sec:truncation}, 
there exists a distribution where deterministically for \emph{any} set of 
up to $o(d)$ many uncorrupted samples, \emph{no} large subset can be stable.
This distribution is the one returning a vector of length $\sqrt{d}$ 
from a randomly chosen axis direction, which has unit covariance.
Essentially, the long length of $\sqrt{d}$ along directions as sparse as the axis directions
causes stability to fail to hold.

In order to circumvent this obstacle, we propose to ``truncate'' all the samples 
in $\ell_\infty$ norm before using a stability-based filtering robust mean estimation algorithm.
Specifically, we start by computing an initial mean estimate, 
and then clip each sample coordinate-wise 
to within a radius of $\Theta(\sqrt{k})$ of the initial mean estimate. 
This radius is chosen carefully to ensure that the mean of the original distribution 
and the clipped distribution is close in $\ell_{2,k}$-norm. 
Ensuring that the clipped distribution also has small variance 
turns out to be non-trivial, as we detail below.  

\paragraph{Necessity of Bounded Higher Moments} 
After truncation, we have the guarantee that no point is too far from the true mean. 
Unfortunately, truncation can potentially also \emph{rotate} a point about the true mean, 
in the sense that for a sample, the direction of its difference 
from the true mean may change after such truncation.
In general, this rotation effect can cause much of the mass of the distribution 
to rotate and concentrate towards certain directions, 
and significantly \emph{increase} the variance in those directions.
(See \Cref{app:truncation} for more details.)
In this work, we identify the mild condition that the $4^\text{th}$ moment 
is bounded along each {\em axis} direction by some constant, 
on top of our assumption that $\Sigma \preceq I$, 
to be sufficient to show that truncation can only increase variance in directions 
that are non-sparse --- in the sense that the resulting covariance 
will still have bounded $\Chi_k$ norm (see \Cref{lem:truncationInfBd}).
Thus, under these mild conditions, we can safely truncate our samples 
(which is necessary for stability to hold, as outlined above), 
and modify our goal to show this truncated distribution 
contains a large stable set with high probability.

\paragraph{Stability of Truncated Samples with High Probability}

Even after truncation and after imposing an axis-wise $4^\text{th}$ moment bound, 
it remains challenging to show that, with high probability, there exists 
a large subset of samples that are stable with respect to the true mean.

As we see in \Cref{sec:stability_after_removing_points_additive_dependence_on_}, 
the analysis reduces to showing that with high probability over the uncorrupted samples, 
for every matrix $M \in \Chi_k$, there exists a large subset of samples $S$ 
whose empirical covariance $\overline{\Sigma}_S$ has a small inner product with $M$, 
namely that $ M \bullet \overline{\Sigma}_S $ is bounded.
In the non-sparse setting, the strategy used in~\cite{DepLec22} and~\cite{DiaKP20} 
is to first show a high probability event for all $M = vv^\top$ for unit vectors $v$, 
and then to show that the event for all $M = vv^\top$ deterministically implies 
that the event holds also for all $M \succeq 0$ with $\tr(M) = 1$.
{This strategy is important because although the cover of PSD matrices 
would roughly be exponential in $d^2$, the cover of $vv^\top$ is only exponential in $d$.}
{Thus, the first step holds with roughly $d$ samples,}
and the second step crucially uses the spectral decomposition (SVD) 
of positive semidefinite (PSD) matrices. 
On the other hand, in our sparse setting, if we applied the usual SVD 
to the PSD matrices $M \in \Chi_k$, the resulting decomposition 
will generally not yield sparse components,
and thus not allowing us to leverage sparsity.
Instead, inspired by certain matrix norm results derived by Li~\cite{li18thesis}, 
we carefully design a (non-spectral) decomposition 
that does yield {$k^2$-}sparse 
components and can be covered {with $k^2\log d$ samples};
as well as a more delicate argument to complete the second step, 
namely that the event holding for all components $M$ 
in the decomposition implies the event holding for all $M \in \Chi_k$.
The intricacies of these arguments also allow us 
to get a sample complexity that ultimately yields 
an {\em additive} (as opposed to multiplicative) dependence on $\log 1/\tau$, 
which as described in the previous section is a crucial feature of our result, 
and in line with the non-robust non-sparse sub-Gaussian mean estimation setting.

\subsection{Related Work}
\label{sec:literature}

\paragraph{Algorithmic Robust Statistics}
The goal of algorithmic robust statistics is to obtain dimension-independent 
asymptotic error even in the presence of constant fraction 
of outliers in high dimensions in a computationally efficient way.
Since the dissemination of \cite{DKKLMS16,LaiRV16}, 
which focused on high-dimensional robust mean estimation, 
the body of work in the field has grown rapidly.
For example, prior work has obtained dimension-independent 
guarantees for various problems such 
as linear regression~\cite{KlivansKM18,DKS19-lr} and convex optimization~\cite{PraSBR20,DiakonikolasKKLSS2018sever}. 
See the survey \cite{DK19-survey} for a more detailed description.
Most relevant to us are the works on robust mean estimation 
that leverage the sparsity constraints and obtain improved sample complexity. 
The algorithms developed in \cite{BDLS17,DKKPS19-sparse,CDKGS21-sparse,DKKPP22-sparse-sos} obtain optimal asymptotic error for light-tailed distributions, such as Gaussians.
However, these algorithms (and their analyses) crucially rely 
on the light-tails and, as outlined in \Cref{sec:our-approach}, 
provably do not work for heavy-tailed distributions.

\paragraph{Heavy-Tailed Statistical Estimation}
The recent decades also saw a growing interest 
in studying statistics in heavy-tailed settings.
Even for the basic question of univariate mean estimation without sample corruption, 
the statistical limits are only recently resolved by 
a line of work started by Catoni~\cite{Cat12} and ending 
with Lee and Valiant~\cite{LeeValiant2020} 
(see also~\cite{Minsker:2022-subgaussian-mean} for an alternative estimator).

The high-dimensional heavy-tailed setting turns out to be much more challenging
and has been extensively studied in recent years, e.g., 
for mean estimation in the $\ell_2$ norm~\cite{LugMen19-mean} 
and in other norms~\cite{LugMen19-gen,DepLec22-pseudo}, 
covariance estimation~\cite{MenZhi20}, and stochastic convex optimization~\cite{BarMen22}.
In absence of contamination, the goal is to obtain sample complexity 
as if the distribution were Gaussian.
Roughly speaking, this corresponds to an additive dependence 
on the logarithm of failure probability 
in various estimation tasks (as we achieve also in this work).
We refer the reader to the survey for more details~\cite{LugMen19-survey}.
This line of work focuses on the statistical limits, 
and the estimators developed are generally computationally inefficient.

A closely-related body of research aims to obtain \emph{efficient} 
algorithms for heavy-tailed distributions with optimal statistical performance, 
\new{ideally} matching the above guarantees.
These works include high-dimensional (dense) mean estimation 
\cite{Hop20,CherapanamjeriF19, DepLec22,LeiLVZ20,DiaKP20,HopLZ20,CheTBJ22,LeeValiant2022}, 
linear regression~\cite{CheHRT20,PenJL20,Dep20}, and covariance estimation~\cite{CheHRT20}.
We note that many of these works are inspired by the algorithmic robust statistics 
literature and can also tolerate a constant fraction of contaminated data.

To the best of our knowledge, none of these works studies sparse estimation 
under heavy-tailed distributions (even in absence of outliers), 
and our work is the first result with sample complexity 
that is additive in the logarithm of the failure probability.

\subsection{Organization} \label{ssec:org}
The structure of this paper is as follows: After the necessary technical preliminaries
in Section~\ref{sec:prelim}, in Section~\ref{sec:truncation} we describe and analyze
our simple pre-processing truncation scheme. In Section~\ref{sec:alg_analysis}, we provide a detailed
description of our algorithm and an outline of its analysis, assuming the necessary
stability conditions are satisfied. Sections~\ref{sec:stability_after_removing_points_additive_dependence_on_}
and~\ref{app:stability-smoothness} establish that the stability condition will be satisfied
with the appropriate sample complexity, and are the main technical contributions of this work.
Finally, Section~\ref{sec:information_theoretic_error} shows that the error guarantee 
of our algorithm is information-theoretically optimal under a mild assumption on the sparsity.
For the sake of the presentation, some technical proofs have been deferred to an appendix.

\section{Preliminaries}
\label{sec:prelim}

\paragraph{Notations}
Here we define the notations we use in the rest of the paper.
For a (multi-)set $S \subset \R^d$, we denote  $\mu_S = (1/|S|)\sum_{x \in S}x$ and $\Sigma_S = (1/|S|) (\sum_{x \in S}  (x - \mu_S)(x - \mu_S)^\top)$. 
When the vector $\mu$ notation is clear from context, we use $\overline{\Sigma}_S$ to denote $(1/|S|) \sum_{x \in S} (x - \mu)(x - \mu)^\top$.

Let $\cU_k$ denote the set of $k$-sparse unit vectors in $\R^d$. For two vectors $x$ and $y$, $\langle x, y\rangle$ denotes the dot product $x ^\top y$.
For a vector $x \in \R^d$, we use $\|x\|_{2,k}:= \sup_{v \in \cU_k} \langle x,v\rangle$ and $\|x\|_{\infty}$ to denote $\max_j |x_j|$.
For a matrix $M$, we use $\|M\|_1$ to denote $\sum_{i,j}|M_{i,j}|$ and $\|M\|_0$ to denote the number of non-zero entries of $M$.
For two matrices $A$ and $B$, we use $ A \bullet B $ to denote the trace inner product $\trace(A^\top B)$.
Define $\cX_k := \{M: M \succcurlyeq 0, \trace(M) = 1 , \|M\|_1 \leq k\}$. For a matrix $A$, we define $\|A\|_{\cX_k}:= \sup_{M \in \cX_k} |  A \bullet  M |$.
For an $n \in \N$, we use $[n]$ to denote the set $\{1,\ldots,n\}$.
For a set $S \subseteq \R^d$ and a function $f$, we also define the set function notation $f(S)$ as $\{f(x) \, | \, x \in S\}$.

\paragraph{Coordinate-wise Median-of-Means}
We use the coordinate-wise median-of-means algorithm to robustly obtain a preliminary mean estimate, with guarantees captured by the following fact.

\begin{fact}
\label{fact:co-MoM}
The coordinate-wise median-of-means algorithm satisfies the following guarantee: given the corruption parameter $\eps$, failure probability $\tau$, and a set $T$ of $n$ many $\eps$-corrupted samples from a distribution $D$ with mean $\mu$ and axis-wise variance $\Exp_{X \sim D}[(X_j - \mu_j)^2] \le \sigma^2$ for all $j \in [d]$, then with probability at least $1-\tau$ over the sample set $T$, the output of the algorithm $\muhat$ is such that $\|\muhat - \mu\|_{\infty} \le \sigma O(\sqrt{\eps} + \sqrt{(\log (d/\tau))/n})$.
\end{fact}

\paragraph{Median-of-Means Pre-Processing}
 Another standard technique we use in this paper is the median-of-means pre-processing, which is a distinct technique from the coordinate-wise median-of-means algorithm mentioned right above.
Recall that in \Cref{thm:ourResult}, the asymptotic error term is $\sqrt{\eps}$, which tends to 0 as the corruption parameter $\eps \to 0$.
The following pre-processing step allows us to reduce the problem from the $\eps \to 0$ case to a constant $\eps$ case:
Let $T$ be the input $\eps$-corrupted set of samples.
Split the samples $T$ randomly into $g$ equally-sized groups of size $m = n/g$ where $g = 100 \eps n$, and replace each group by the sample mean of the group.
Let $T_{\mathrm{grouped}}$ be this new set of points.
It is easy to check that at most $0.01$-fraction of $T_{\mathrm{grouped}}$ can be corrupted by outliers.
The effects of this pre-processing is captured by the following \Cref{fact:MoM}, which we prove for completeness in \Cref{sec:fact-mom}.

\begin{restatable}[Median-of-Means Pre-Processing]{fact}{FactMoMPreProcessing}
\label{fact:MoM}
Suppose there is an efficient algorithm such that, on input {$\sigma \in \R_+$} and a $0.01$-corrupted set of $n \gg k^2\log d + \log(1/\tau)$ samples from a distribution $D$ with mean $\mu$ and covariance $\Sigma$ with $\|\Sigma\|_{\Chi_k} \leq \sigma^2$ and $\Exp_{X \sim D}[(X_j - \mu_j)^4] = O(\sigma^4)$ for each coordinate $j \in [d]$, returns $\muhat$ such that $\|\muhat - \mu\|_{2,k} \le O(\sigma)$ with probability at least $1-\tau$.

Then, there is an efficient algorithm such that, on {input $\eps \in (0,0.01] $} and an $\eps$-corrupted set of $n \gg (k^2\log d + \log (1/\tau))/\eps$ samples from a distribution with mean $\mu$ and covariance $\Sigma$, satisfying (i)  $\|\Sigma\|_{\Chi_k}  \leq 1$ and (ii) $\Exp_{X \sim D}[(X_j - \mu_j)^4] = O(1)$ for every coordinate $j \in [d]$, returns a mean estimate $\muhat$ such that $\|\muhat - \mu\|_{2,k} \le O(\sqrt{\eps})$ with probability at least $1-\tau$.
\end{restatable}

\section{Truncation Pre-Processing} %
\label{sec:truncation}

The general approach of using a stability-based filtering algorithm 
for robust mean estimation is to show that, given sufficiently many samples, 
there exists a large (say $1-O(\eps)$ fraction) subset of the samples 
that are stable with respect to the true mean $\mu$.
However, the following simple example shows that it is not possible 
for i.i.d.\ samples drawn from a heavy-tailed distribution 
to satisfy the sparse stability condition using sample size of $\poly(\log d)$. 
\begin{restatable}{example}{ExmTruncNecessary}
\label{ex:HeavyTailed}
For any number of moments $t \ge 2$, there is a distribution $X$ satisfying the following conditions:
  (i)
  The mean of $X$ is $0$, and for every unit vector $v$, the $t^\text{th}$ moment in direction $v$ is upper bounded by 1, that is, $\Exp[|\langle v,  x\rangle |^t] \le 1$ for $t \geq 2$,
  (ii)
  If $S$ is an {arbitrary} set of $n \le o(d^{2/t})$ points from the support of $X$, then the set $S$ cannot be $(\epsilon,O(\sqrt{\epsilon}),k)$-stable, for any $\eps>0$, with respect to the mean of the distribution.
  As a corollary, no subset of $S$ can be stable either.
\end{restatable}
\begin{proof} 
For $j \in [d]$, let $e_j$ be the vector that is $1$ on the $j$-th coordinate and $0$ otherwise.
For a fixed $r$, consider the distribution $P$, supported uniformly on the $2d$ points $S = \{ \pm re_1,\pm re_2,\dots, \pm re_d \} $.

It follows that $P$ is a zero mean distribution. The covariance of the distribution $P$ is $\sum_j(1/d)r^2 e_ie_i^\top  = (r^2/d)I$. Furthermore, for any unit vector $v$ and $t \geq 2$, 
we have that the $t$-th moment in the direction $v$ is bounded as follows:
\begin{align*}
    \E[ |v \cdot X|^{t}] = \sum _{j=1}^d \frac{1}{d} |v_j|^{t} r^{t} = \frac{r^t}{d} \|v\|_t^t \leq \frac{r^t}{d} \|v\|_2^t \leq \frac{r^t}{d} \;,
\end{align*}
where we use that $t \geq 2$ and $\|v\|_t \leq \|v\|_2$ for any vector $v$.
Thus, we choose $r = d^{1/t}$ for the distribution.

Now we show the second claim, that \emph{any} set of at most $\Omega(d^{2/t})$ samples from this distribution cannot be stable.
Let $S$ be any (multi-)set of $n$ points from the support of $X$.
Let $x_1 \in S$. Since $x_1$ is $1$-sparse and has $\ell_2$ norm $r$, 
we have that $x_1x_1^\top /r^2$ belongs to $\Chi_k$. 
Thus, we have the following: 
\begin{align*}
    \left\|\frac{1}{n}\sum_{i \in S'} x_ix_i^\top  \right \|_{\Chi_k} \geq \left\langle \frac{1}{n}\sum_{i \in S'} x_ix_i^\top  , \frac{1}{r^2} x_1x_1^\top \right\rangle \geq \frac{\|x_1\|^4}{r^2n} = \frac{r^2}{n} \;.
\end{align*}
Therefore, for $r^2/n$ to be upper bounded by a constant, $n$ has to be $\Omega(d^{2/t})$.
\end{proof}

Consequently, if we want to perform robust sparse mean estimation 
using $\poly(k,\log d)$ samples, we need to modify the algorithm.
Our approach is to perform an initial truncation of the samples 
before using a stability-based robust mean estimator.
A balance needs to be struck in order to truncate sufficiently 
for stability to hold (with high probability over the samples), 
but also to truncate mildly enough such that the mean (and covariance) 
of the truncated distribution does not shift too much.

For a scalar $a \in \R_+$ and a vector $b \in \R^d$, let $h_{a,b} : \R^d \to \R^d$ be the following thresholding function:
    \begin{align}
    \label{eq:hDefn}
\forall i \in [d], \,\,\,\,\,\,\,    h_{a, b}(x)_i = \begin{cases} x_i, & \text{ if } |x_i - b_i| \leq a \\                    b_i + a& \text{ if } x_i - b_i \geq a\\
    b_i - a& \text{ if } x_i - b_i \leq - a\end{cases}.
    \end{align}
Note that $h_{a,b}(x)$ projects the point $x$ to the $\ell_\infty$ ball of radius $a$ around $b$.

As explained in the Introduction, truncation in general rotates a point about the true mean, and thus can in fact cause the covariance of the distribution to grow in certain directions.
The following lemma captures the fact that, if we make the further mild assumption that the distribution has bounded $4^\text{th}$ moment along all the axis directions, then we will at least be able to preserve the $\Chi_k$ norm of the covariance matrix.
The proof of \Cref{lem:truncationInfBd} is in \Cref{app:TruncPreserve}.

\begin{restatable}[Truncation in $\ell_\infty$]{lemma}{LemTruncLInftyNorm} 
\label{lem:truncationInfBd}
Let $P$ be a distribution over $\R^d$ with mean $\mu_P$ and covariance $\Sigma_P$, with $\|\Sigma\|_{\Chi_k} \leq \sigma^2$ for some $\sigma^2 > 0$. Let $X \sim P$ and assume that for all $j \in [d]$, $\E[(X - \mu_P)_j^4] \leq \sigma^4 \nu^4 $ for some $\nu \geq 1$.
Let $b \in \R^d$ be such that $\|b - \mu\|_\infty \leq a/2$ and $a := 2\sigma\sqrt{k/\eps}  $ for some $\eps \in (0,1)$.
Define $Q$ to be the distribution of $Y:= h_{a,b}(X)$. Let the mean and covariance of $Q$ be $\mu_Q$ and $\Sigma_Q$ respectively.
Then the following hold:
\begin{enumerate}[label=(\arabic*)]
   \item $\|\mu_P - \mu_Q\|_{\infty} \leq \sigma \sqrt{\epsilon/k}$
   \item $\|\mu_P - \mu_Q\|_{2, k} \leq \sigma \sqrt{\epsilon}$
   \item $\|\Sigma_P - \Sigma_Q\|_{\Chi_k} \leq 3 \sigma^2 \eps \nu^4$
   \item For all $i\in [d]$, $\E[ (Y - \mu_Q)_i^4] \leq 8 \nu^4 \sigma^4$
   \item $\|Y - \mu_Q\|_\infty \leq 2a = 4\sigma\sqrt{k/\epsilon}$ almost surely.
 \end{enumerate} 
\end{restatable}
In \Cref{lem:truncationInfBd} above, $b$ represents the initial mean estimate, and $\tilde{\mu}$  will be obtained by \Cref{fact:co-MoM}. 

\section{Algorithm and Analysis}
\label{sec:alg_analysis}

The high-level algorithm we propose is stated as follows.

\begin{algorithm}[H]
\caption{Robust Sparse Mean Estimation with High Probability}
\label{alg:main}
\begin{enumerate}
    \item Input: An $\eps$-corrupted sample set $T \subseteq \Real^d$ of size $n$

    \item Median-of-Means pre-processing: Group points in $T$ into $g$ groups, each of size $m = n/g$, where $g =  100 \eps n$, and take the sample mean of a group to be a new point. Call these new points $T_{\mathrm{grouped}}$.
    \item Define $\sigma = 1/\sqrt{m}$.
    \item Compute coordinate-wise median-of-means estimate $\tilde{\mu}$ from \Cref{fact:co-MoM} with corruption parameter $0.01$ and failure probability $\tau/2$, using the set of points $T_{\mathrm{grouped}}$. %
    \item Truncate all points in $T_{\mathrm{grouped}}$ to within $B_{\infty}(\tilde{\mu}, 4\sigma\sqrt{k})$, namely, given a point $x$, we replace it with the point $h_{4\sigma\sqrt{k},\tilde{\mu}}(x)$, where $h_{a,b}$ is defined in \Cref{eq:hDefn}.
    \item Run the stability-based robust \emph{sparse} mean estimator from \Cref{fact:SDPAlg} on the truncated samples, i.e.~$\left\{h_{4\sigma \sqrt{k},\tilde{\mu}}(x) \mid x \in T_{\mathrm{grouped}}\right\}$.
\end{enumerate}
\end{algorithm}

We note that this algorithm is shift and scale invariant, 
based on the same invariance of the median-of-means pre-processing 
as well as the invariance of the robust sparse mean estimator from \Cref{fact:SDPAlg}.

We will now prove that \Cref{alg:main} satisfies the guarantees of \Cref{thm:ourResult} 
and its stronger version, \Cref{thm:ourResult_Chi_k} stated below.
Our analysis crucially uses \Cref{thm:finalStable}, which states that, 
with high probability, there exists a set consisting of most 
of the truncated samples that is stable with respect to some point 
close to the true mean in $\ell_\infty$ norm.
This is the main structural result of our paper.

\begin{restatable}{theorem}{ThmFinalStable}
\label{thm:finalStable}
Let $S$ be a set of $n$ i.i.d.\ data points from a distribution $P$ over $\R^d$, 
and let $T$ be a $0.01$-corruption of $S$.
Let $\tilde{\mu}$ be the coordinate-wise median-of-means estimate computed from set $T$.
Let the mean of $P$ be $\mu$ and covariance $\Sigma$ such that $\|\Sigma\|_{\Chi_k} \le \sigma^2$, and for all $i \in [d]$, $\E[X_i^4] \le O(\sigma^4)$.
Suppose that $n = \Omega(k^2 \log d + \log(1/\tau))$. 
Let $a = \sigma\sqrt{k}$.
With probability $1- \tau$ over $S$, for all $T$ 
we have that there exists a subset $S' \subseteq T$ 
with $|S'| \ge 0.95 n$ such that $h_{a,\tilde{\mu}}(S')$ 
is $(0.01, O(1) ,k)$-stable with respect to some $\mu'$ 
and $\sigma$ with $\|\mu'-\mu\|_{\infty} \le O(\sigma/\sqrt{k})$.
\end{restatable}

We show \Cref{thm:finalStable} in two steps.
First, we show a simpler, analogous stability result assuming that we truncate with respect to the true mean vector $\mu$ instead of the coordinate-wise median-of-means estimate $\tilde{\mu}$ as in \Cref{alg:main} and \Cref{thm:finalStable}.
This is stated as \Cref{thm:stabHighProb} and proved in \Cref{sec:stability_after_removing_points_additive_dependence_on_}.
Then, in \Cref{app:stability-smoothness}, we show a ``Lipschitzness'' argument that lets us conclude \Cref{thm:finalStable}.
The final proof of \Cref{thm:finalStable} is given in \Cref{sec:proofoffinalstable}.

\begin{theorem}[Main Result, Strong Version]
\label{thm:ourResult_Chi_k}
 Let $\epsilon \in (0, \epsilon_0)$ for small constant $\epsilon_0>0$. 
 Let $P$ be a multivariate distribution over $\R^d$, 
 where the mean and covariance of $P$ are $\mu$ and $\Sigma$ respectively.
 Suppose $\|\Sigma\|_{\Chi_k} \le 1$ and further suppose that 
 for all $j \in [d]$, $\E[(X_j-\mu_j)^4] = O(1)$.
  Then, there is an algorithm such that, on input 
  (i) the corruption parameter $\epsilon$, 
  (ii) the failure probability $\tau$, 
  (iii) the sparsity parameter $k$, and 
  (iv) $T$, an $\eps$-corrupted set of  
   $n \gg  (k^2\log d + \log(1/\tau))/\eps)$ i.i.d.~samples from $P$, 
   it outputs $\widehat{\mu}$ satisfying 
   $\|\widehat{\mu} -  \mu\|_{2,k}$ $= O(\sqrt{\eps})$ 
   with probability $1 - \tau$ in $\poly(n,d)$ time.
\end{theorem}

We note that, since the $\Chi_k$ norm of a covariance matrix is upper bounded by its maximum eigenvalue, \Cref{thm:ourResult} is an immediate corollary of \Cref{thm:ourResult_Chi_k}.

\begin{proof}[Proof of \Cref{thm:ourResult_Chi_k}]
Step 2 of \Cref{alg:main}, the median-of-means pre-processing, is exactly the same as the reduction in \Cref{fact:MoM}.
Thus, by \Cref{fact:MoM}, it suffices to show that, for every $\sigma > 0$, 
Steps 4--6 in \Cref{alg:main} yield an $O(\sigma)$ estimation error 
in $\ell_{2,k}$ norm when given $0.01$-corrupted samples from a distribution 
$D$ with covariance bounded by $\sigma^2$ in $\Chi_k$-norm 
and axis-wise 4th moment bounded by $O(\sigma^4)$.

\Cref{thm:finalStable} states that, with probability at least $1-\tau$, the samples after the processing of Step 5 are such that there exists a 95\% of the samples that form a $(0.1, O(1),k)$-stable subset with respect to some vector $\mu'$ and $\sigma$ with $\|\mu'-\mu\|_{\infty} \le O(\sigma/\sqrt{k})$.
\Cref{fact:SDPAlg} then guarantees that, on input such a set of samples, 
the routine we invoke in Step 6 of \Cref{alg:main} will return a mean 
estimate $\muhat$ such that $\|\muhat - \mu'\|_{2,k} \le O(\sigma)$.
Further, since $\|\mu'-\mu\|_{\infty} \le O(\sigma/\sqrt{k})$, 
we have that $\|\mu'-\mu\|_{2,k} \le O(\sigma)$, and therefore 
we can conclude via the triangle inequality 
that the mean estimate $\muhat$ satisfies $\|\muhat-\mu\|_{2,k} \le O(\sigma)$.
\end{proof}

\section{Stability After Removing Points: Additive dependence on $\log(1/\tau)$} %
\label{sec:stability_after_removing_points_additive_dependence_on_}

In this section, we give the core part of the argument (\Cref{thm:stabHighProb}) for the main stability result (\Cref{thm:finalStable}) in this paper.
Recall, via the median-of-means pre-processing, that we only need to consider the constant contamination case $(\eps = \Theta(1))$.
Thus, the goal is to show (\Cref{thm:finalStable}) that with high probability, after truncation according to the coordinate-wise median-of-means preliminary estimate, there exists a large subset of uncontaminated samples that is $(\Theta(1), O(1),k)$-stable with respect to (a vector close in $\ell_{2,k}$ norm of) the \emph{true mean} of the distribution as well as $\sigma$ where $\sigma^2 = \|\Sigma\|_{\Chi_k}$.

The key difference between Theorems~\ref{thm:stabHighProb} and~\ref{thm:finalStable} is that the former is a stability result that applies only to uncontaminated i.i.d.~samples truncated according to some \emph{fixed} vector close to the true mean.
On the other hand, the final stability result we require concerns samples truncated according to the coordinate-wise median-of-means estimate, which itself depends on the samples and is not fixed.
\Cref{app:stability-smoothness} shows the delicate argument going from \Cref{thm:stabHighProb} to \Cref{thm:finalStable}.

\begin{theorem} \label{thm:stabHighProb}
Let $S$ be a set of $n$ i.i.d.\ data points from a distribution $P$ over $\R^d$. 
Let the mean of $P$ be  $\mu$ and covariance $\Sigma$ 
such that $\|\Sigma\|_{\Chi_k} \le \sigma^2$, and for all $j \in [d]$, $\E[(X_j - \mu)^4] \le \nu^4$.
Suppose $P$ is supported over the set $ \{x: \|x - \mu\|_\infty \leq \sigma \times \InftyBd \times  \sqrt{k}\}$.
If $n = \Omega(k^2\log d + \log(1/\tau))$, then with probability $1 - \tau$ 
there exists a set $S' \subseteq S$ such that:
\begin{enumerate}
  \item $|S'| \geq 0.98 n$
  \item $S'$ is $(0.01, \delta ,k)$-stable with respect to $\mu$ and $\sigma$ where $\delta = O(\max(1, r^2, \nu^2/\sigma^2))$.
\end{enumerate}
\end{theorem}

Before giving the proof, we point out that the specific application 
of \Cref{thm:stabHighProb} will be on the samples $h_{a,\mu}(x_i)$, 
where $x_i$ are the original uncontaminated i.i.d.\ samples, 
$\mu$ is the underlying mean vector we are trying to estimate, 
and for $a$ appropriately chosen to match \Cref{alg:main}.

\begin{proof}
In the following proof, we will use notations $\AKPprob,\AKPtail,\Qconst,\Qscale,V_Z$ and $\res$, 
all of which are either constants or functions of $\sigma$, $r$ and $\nu$ in the theorem statement.
The functions are explicitly chosen in \Cref{app:consts}.

We will assume $\mu = 0$ without loss of generality.
Instead of directly showing the existence of subset $S' \subseteq S$ 
(with high probability over the samples $S$) that is stable, 
\Cref{Prop:stabilityMoMCov} in \Cref{app:misc} lets us show the following simpler condition: 
let $\Delta_{n,\eps}$ be the set of weights/distributions $w$ such that $w_i \le 1/(1-\eps)$, 
then there exists a weighting $w \in \Delta_{n,0.01}$ such that $\|\Sigma_w\|_{\Chi_k} \le \res$ 
for the function $\res$ chosen in \Cref{app:consts}, 
which satisfies $\res = O(\sigma^2 \max(1, r^2, \nu^2/\sigma^2))$.
That is, for the following proof, we just need to prove that 
$\min_{w \in \Delta_{n,0.01}}\|\Sigma_w\|_ {\Chi_k} \le \res$.

We proceed as follows: 
\begin{align*}
\min_{w \in \Delta_{n,0.01}}\|\Sigma_w\|_ {\Chi_k} 
&=  \min_{w \in \Delta_{n,0.01}} \max_{M \in \Chi_k} \left\langle M ,\Sigma_w \right\rangle 
= \max_{M \in \Chi_k} \min_{w_M \in \Delta_{n,0.01}}  \left\langle M ,\Sigma_w \right\rangle \;,
\end{align*}
where the last equality is a straightforward application of the minimax theorem for a minimax optimization problem with independent convex domains and a bilinear objective.
It thus suffices to show the following: with probability $1 - \tau$,
\begin{align}
\label{eq:UnifConvChik}
\forall M \in \Chi_k: \,\, |\{x \in S: x_i^\top Mx_i > \res\}| \leq 0.01|S|
\end{align}
from which we can construct the weighting $w_M$ as uniform distribution over the elements outside the above set.

Define the following sets of sparse matrices:
\begin{align}
\label{eq:DefAkp}
\cA_k := \left\{ A \in \R^{d \times d}: \|A\|_0 \leq k^2, \|A\|_F \leq 1   \right\} 
\nonumber\\
\cA_{k,P} := \left\{ A \in \cA_k: \P\left\{ x^\top Ax \geq \AKPtail \right\} \leq \AKPprob \right\} \;,
\end{align}
where $\AKPprob$ and $\AKPtail$ are chosen in \Cref{app:consts}.
If $n \gtrsim (k^2\log d + \log(1/\tau))/(\AKPprob^2)$, 
then a standard covering/VC-dimension bound (see \Cref{lem:unifConcAkp} for details) 
implies that the following event holds with probability $1 - \tau$:
\begin{align}
\label{eq:UnifConvAkprime}
\forall A \in \cA_{k,P}: \,\, |\{x \in S: x_i^\top Ax_i > \AKPtail\}| \leq 2\times\AKPprob \cdot |S| \;.
\end{align}
Our choice of $\AKPprob$ is a constant (cf.~\Cref{app:consts}) 
and thus the required sample complexity for \Cref{eq:UnifConvAkprime} 
to hold is $\Omega(k^2\log d + \log(1/\tau))$.
We will now show that the event in \Cref{eq:UnifConvAkprime} 
implies that the event in \Cref{eq:UnifConvChik} holds.

Suppose, for the sake of contradiction, that the event in \Cref{eq:UnifConvChik} does not hold.
Then there exists an $M \in \Chi_k$ such that $|\{x \in S: x_i^\top Mx_i > \res\}| > 0.01|S|$.
We will show the existence of a matrix $Q$ violating event~\Cref{eq:UnifConvAkprime}, 
via the probabilistic method, to reach the desired contradiction.

Fixing an $M$ that violates \Cref{eq:UnifConvChik}, 
consider the random matrix $Q$ where each entry  $Q_{i,j}$ is sampled independently 
from the following distribution, defined using the constant $\Qconst$ chosen in \Cref{app:consts}:
\begin{align}
Q_{i,j} : = \begin{cases} M_{i,j}, &\text{ with prob. 1 if } |M_{i,j}| \geq \Qconst/k, \\
   \frac{\Qconst}{k}\text{sign}(M_{i,j}), &\text{ with prob.   $|kM_{i,j}|/\Qconst$ if $|M_{i,j}| \leq \Qconst/k$},\\
   0,  &\text{ with remaining prob. if $|M_{i,j}| \leq \Qconst/k$}
\end{cases}.
\label{eq:Qsampling}
\end{align}
Defining $p_{i,j}$ to be $\min(1,k|M_{i,j}|/\Qconst)$, 
then $Q_{i,j}$ is equivalently $M_{i,j}/p_{i,j}$ with probability $p_{i,j}$ and $0$ otherwise.

We will show that the following events hold simultaneously 
with non-zero probability, leading to a contradiction to event~\Cref{eq:UnifConvAkprime}:
\begin{enumerate}[label=(\Roman*)]
  \item \label{item:QNice} $Q \in \Qscale\cA_{k,P}$,
  \item \label{item:QBad}$|\{x \in S: x_i^\top Qx_i > \Qscale\times\AKPtail\}| > 2\times \AKPprob \cdot |S|$,
\end{enumerate}
where $\Qscale$ is also a constant, larger than $2$, and explicitly chosen in \Cref{app:consts}.
Using different techniques, we will show that the first condition holds with probability 
at least $1 - 2 \times 10^{-6}$ and the second condition holds with probability 
at least $4 \times 10^{-6}$, thus implying that the events 
hold simultaneously with non-zero probability.

\paragraph{Condition~\ref{item:QNice}} %

Showing that $Q$ belongs to $s_3\cA_k$ with high constant probability is straightforward: by the construction of $Q$, it has small expected sparsity as well as small expected Frobenius norm.
An application of Markov's inequality shows that $Q \in s_3\cA_k$ with high constant probability (\Cref{lem:QinA_k}).

The trickier part is to show that $Q$ is also in $s_3\cA_{k,P}$, namely that $\Pr_{x \sim P}(x^\top  Q x > s_3\times s_1)$ is upper bounded by the small constant.
We consider the distribution of $x^\top Q x $ over the probability of independently drawing $x \sim P$ and a random $Q$, and show that $x^\top  Q x $ is small with high probability over this joint distribution (\Cref{lem:Q-in_akp}), which requires using the axis-wise $4^\text{th}$ moment bounds on $P$ as well as the fact that $M \in \Chi_k$.
\Cref{lem:jointProbToconditioned} implies that with high probability, we will draw a $Q$ satisfying $\Pr_{x \sim P}( x^\top  Q x > s_3 \times s_1)$ being bounded by a small constant.

We now show Condition~\ref{item:QNice} as sketched above, beginning with the following lemma showing that $Q$ lies in $\Qscale\cA_k$ with high probability.
\begin{lemma}[$Q$ lies in $\Qscale\cA_k$ with high probability]
\label{lem:QinA_k}
Let $Q$ be generated as described in \Cref{eq:Qsampling}, for an $M \in \Chi_k$. Then with probability except $(1/ \Qconst) + (\Qconst/\Qscale^2)$, we have that $Q \in \Qscale \cA_k$.
\end{lemma}

\begin{proof}
The expected sparsity of $Q$ is at most $\sum_{i,j} \frac{k}{\Qconst}|M_{i,j}| \le \frac{k^2}{\Qconst}$ since $|M|_1 \le k$.
Thus, by Markov's inequality, except with $1/\Qconst$ probability, $Q$ and hence $Q/\Qscale$ is $k^2$-sparse.
We also have to show that with probability at least $1-10^{8}$, $\|Q\|_F \le \Qscale$.
\begin{align}
    \E\|Q\|_F^2 \leq \sum_{i,j} \left(\frac{\Qconst}{k} \right)^2 \left(\frac{k |M_{i,j}|}{\Qconst}\right) =   \frac{\Qconst|M_{i,j}|}{k} = \Qconst.
\end{align}
Again, by Markov's inequality, we get that with probability except $\Qconst/\Qscale^2$, the Frobenius norm of $Q$ is at most $\Qscale$.
The lemma statement follows from the union bound.
\end{proof}
Our choice of constants in \Cref{app:consts} would ensure that the failure probability in \Cref{lem:QinA_k} is at most $10^{-6}$. That is,
\begin{align}
\label{eq:QinAk}
    \frac{1}{\Qconst} + \frac{\Qconst}{\Qscale^2} \leq 10^{-6}.
\end{align}

It remains to show that $Q$ belongs to $\Qscale \cA_{k,P}$ with high (constant) probability, i.e., with probability $10^{-6}$ over sampling of $Q$, we have that $\P_{x \sim P}(x^\top Qx > \Qscale\times \AKPtail| Q) \leq \AKPprob$. 
Let $R := x^\top Qx$, where both $x$ and $Q$ are sampled independently from $P$ and \Cref{eq:Qsampling} respectively.

To show this, we use the following the lemma for a sufficient condition involving sampling both $x$ and $Q$.

\begin{lemma}
\label{lem:jointProbToconditioned}
Consider a probability space over the randomness of independent variables $X$ and $Y$.
Suppose the event $E$ (over pairs $(X,Y)$) happens with probability at least $1-\alpha\beta$ for some $\alpha, \beta \in [0,1]$.
Then, it must be the case that, with probability at least $1-\alpha$ over the sampling of $X$, the conditional probability of $E$ given $X$ is at least $1-\beta$.
\end{lemma}

\begin{proof}
For the sake of contradiction, suppose the lemma conclusion is false. Then
\begin{align*}
    \Pr_{X,Y}(E) &= \int \Pr_Y(E|X) \, \d \Pr(X)
    < (1-\alpha) + \alpha(1-\beta) = 1-\alpha\beta,
\end{align*}
which contradicts the premise.
\end{proof}

To conclude that $Q \in s_3\cA_{k,P}$ with high probability, 
it thus suffices to show that with probability $1-10^{-6}\times \AKPprob$ 
over both $x$ and $Q$, $R \le \Qscale \times \AKPtail$.

{\begin{lemma} \label{lem:Q-in_akp}
Let $R = x^\top Qx$, where $Q$ is independently drawn from the distribution 
in \Cref{eq:Qsampling} and $x$ is drawn independently from $P$. 
Under the assumptions of \Cref{thm:stabHighProb}, 
\begin{align}
\label{eq:AkpProbCondition-1}
\P\{R > \Qscale\times\AKPtail\}    \leq \frac{\sigma^2}{\AKPtail} + \frac{4}{\Qscale} + \frac{\Qconst \times \nu^4}{ \Qscale \times \AKPtail^2} \;.
\end{align}
\end{lemma}}
\begin{proof}
We consider three exhaustive events, over $x$ and $Q$, of $\cE := \{R > \Qscale\times\AKPtail\}$, and bound the probability of each of them:
\begin{enumerate}
  \item $\cE_1 := \{(x,Q): \E[R|x] > \AKPtail\}$. Since $\E[R|x] = x^\top Mx$, the event corresponds to $\{x: x^\top Mx >  \AKPtail\}$. We have that $\E[x^\top Mx] = \left\langle \Sigma, M \right\rangle \leq \|\Sigma\|_{\Chi_k} = \sigma^2$.
By Markov's inequality, $\P(\cE\cap\cE_1) \leq \P(\cE_1) \leq \sigma^2/(\AKPtail)$.
  \item $\cE_2 :=\{(x,Q): x \in \cF\}$, where $\cF$ is the following event over $x$:  $\cF = \{x: \E[R|x] \leq \AKPtail ,  \Var(R|x) \leq \Qscale \times \AKPtail^2\} $.
  Observe that conditioned on $x \in \cF$,  we have that $R|x$ is a random variable with mean at most $\AKPtail$ and variance at most $\Qscale \times \AKPtail^2$.
  Thus for each such $x \in \cF$, the conditional probability that $R > \Qscale \times \AKPtail $ is at most $\Qscale\AKPtail^2 /((\Qscale-1)^2 \AKPtail^2)$ by Chebyshev's inequality.
  We thus get that $\P(\cE_2 \cap \cE) \leq \P(\cE | \cE_2) =  \P(\cE | x \in \cF)\leq 4/( \Qscale)$, where we use that $\Qscale \geq 2$.

\item $\cE_3 := \{(x,Q): \Var(R|x) \geq \Qscale\times \AKPtail^2\}$. We will upper bound $\P(\cE_3)$.
We first calculate the $\Var(R|x)$ using the independence of entries of $Q$ as follows:
\begin{align*}
\Var(R|x) = \sum_{i,j} x_i^2 x_j^2 \Var(Q_{i,j}) = \sum_{i,j: |M_{i,j}| \leq \Qconst/k} x_i^2x_j^2|M_{i,j}|\left( \frac{\Qconst}{k} - |M_i,j| \right).
\end{align*}
To show that $\Var(R|x)$ is small with high probability, we will upper bound $\E[\Var(R|x)]$.
\begin{align*}
\E[\Var(R|x)] &= \sum_{i,j: |M_{i,j}| \leq \Qconst/k}  |M_{i,j}|\left( \frac{\Qconst}{k} - |M_i,j| \right) \E[x_i^2x_j^2] \\
&\leq \sum_{i,j}  \frac{\Qconst}{k}|M_{i,j}| \E[x_i^2x_j^2] \\
&\leq \frac{\Qconst\times \|M\|_1\times \nu^4}{k} \tag*{(using $\E[x_i^2x_j^2] \le \sqrt{\E[x_i^4]\E[x_j^4]} = \nu^4$)}\\
&\leq \Qconst\times \nu^4 \tag*{(using $\|M\|_1 \leq k$)}.
\end{align*}
Thus Markov's inequality implies that  $\P(\cE \cap \cE_3) \leq \P(\cE_3) \leq (\Qconst \times \nu^4) /  (\Qscale \times \AKPtail^2)$.
\end{enumerate}
Taking the union bound, we get the desired result.
\end{proof}

As reasoned above, we want the failure probability in \Cref{eq:AkpProbCondition-1} to be less than $10^{-6}\times\AKPprob$. That is,
\begin{align}
\label{eq:AkpProbCondition}
\frac{\sigma^2}{\AKPtail} + \frac{4}{\Qscale} + \frac{\Qconst \times \nu^4}{ \Qscale \times \AKPtail^2} \leq 10^{-6}\times\AKPprob.
\end{align}
In \Cref{app:consts}, we choose $\AKPtail, \Qconst, \Qscale$ and $\AKPprob$ such that the bound holds.
This, by the reasoning after \Cref{lem:jointProbToconditioned}, guarantees that $Q$ satisfies the extra condition for $\Qscale\cA_{k,P}$ (on top of being in $\Qscale\cA_{k}$) with probability at least $1-10^{-6}$.

Taking a union bound, with failure probabilities $10^{-6}$ (for $Q$ being in $\Qscale \in \cA_k$, \Cref{lem:QinA_k}) and $10^{-6}$ for satisfying the additional criterion for being in $\Qscale\cA_{k,P}$, we conclude that Condition 1 happens with probability $1-2\cdot10^{-6}$.

\paragraph{Condition~\ref{item:QBad}} %

Define the random variable $Z$ to be
\begin{align}
\label{eq:Zdefinition}
Z = \sum_i \1_{ \left( x_i x^\top_i \right)\bullet Q > \Qscale \times \AKPtail}. 
\end{align}
The second condition is equivalent to saying that $Z > 2 \times  \AKPprob \times |S|$, which we show to happen with probability at least $4 \times 10^{-6}$.

The strategy is to lower bound $\Exp[Z]$, and then use Paley-Zygmund to show that $Z$ is large with constant probability.
To lower bound the expectation, for any $i$ such that $(x_ix_i^\top) \bullet M > \res $, we want to lower bound $\Pr_{Q}((x_i x_i^\top) \bullet Q > \Qscale \times \AKPtail)$, using either Chebyshev's inequality or the Berry-Esseen theorem (see \Cref{fact:berry-esseen}).
First, note that for these $i$, $\Exp[(x_i x_i^\top) \bullet Q] = (x_i x_i^\top) \bullet M > \res$ by our assumption.
If $\Var[(x_i x_i^\top) \bullet Q] \le V_Z$, where $V_Z$ is a fixed function of $ \InftyBd$ and $\sigma$ chosen in \Cref{app:consts}, then by Chebyshev's inequality, we have
\begin{align} 
\label{eq:resVsqscaletail}
\Pr\left((x_i x_i^\top) \bullet Q \ge \Qscale \times \AKPtail \right) \ge \Pr\left((x_i x_i^\top) \bullet Q \ge \res - 10  \times\sqrt{V_Z} \right) \ge 0.99 \;,
\end{align}
where the first inequality is true by our choice of $\AKPtail, \Qscale, V_Z$ and $\res$ in \Cref{app:consts}.
Otherwise, we have the case where $\Var[(x_i x_i^\top) \bullet Q] > V_Z$.
In this case, we treat $(x_i x_i^\top) \bullet Q$ as a sum of independent variables
\begin{align*} (x_i x_i^\top) \bullet Q = \sum_{s, t} (x_i)_s(x_i)_t Q_{s,t} \end{align*} 
and use the Berry-Esseen theorem, which requires bounding the sum of the third central absolute moment of the summands.

\begin{fact}[Berry-Esseen Theorem for Sums of Independent Variables] \label{fact:berry-esseen}
Consider a random variable $\xi = \sum_i \xi_i$, where the variables $\xi_i$ 
are independent (but not necessarily identical) and each of them has finite third moment.
Denote $\mu_i$ as $\Exp[\xi_i]$, $\sigma^2_i$ as $\Var(\xi_i)$ and $\rho_i$ 
as the third central absolute moment, namely $\rho_i = \Exp[|\xi_i-\mu_i|^3]$.
Then,
\begin{align*} 
\DK(\xi, \Normal(\sum_i \mu_i, \sigma^2_i)) \le 0.57 \frac{\sum_i \rho_i}{(\sum_i \sigma^2_i)^{1.5}} 
= 0.57 \frac{\sum_i \rho_i}{(\Var(\xi))^{1.5}} \;,
\end{align*}
where $\DK$ is the Kolmogorov distance between 
two distributions (namely, the $\ell_\infty$ distance between the cumulative density functions).
\end{fact}

Let $\rho_{s,t}$ be the third central absolute moment of $(x_i)_s(x_i)_t Q_{s,t}$. 
For any $(s,t)$ such that $0 < |M_{s,t}| \leq \Qconst /k$, we can calculate its third moment as follows:
\begin{align*}
\rho_{s,t} &= \Exp[|(x_i)_s(x_i)_t Q_{s,t} - \Exp[(x_i)_s(x_i)_t Q_{s,t}] |^3]\\
 &= |(x_i)_s|^3 |(x_i)_t|^3 \frac{|M_{s,t}|^3}{p_{s,t}^3} \Exp[|\Bernoulli(p_{s,t}) - p_{s,t}|^3]\\
 &= |(x_i)_s|^3 |(x_i)_t|^3 \frac{|M_{s,t}|^3}{p_{s,t}^3} p_{s,t}(1-p_{s,t}) (1-2p_{s,t} + 2p_{s,t}^2)\\
 &\le |(x_i)_s|^3 |(x_i)_t|^3 \frac{|M_{s,t}|^3}{p_{s,t}^3} p_{s,t}(1-p_{s,t}) \tag*{\text{for all $p_{s,t} \in [0,1]$}}\\
 &\le (\sigma^2 \times \InftyBd^2 \times \Qconst) (x_i)_s^2 (x_i)_t^2 \frac{M_{s,t}^2}{p_{s,t}^2} p_{s,t}(1-p_{s,t}) \tag*{(\text{since $|x_i|_\infty \le \sigma \times \InftyBd \times \sqrt{k}$ and $|M_{s,t}|/p_{s,t} = \Qconst/k$)} }
\end{align*}
The same inequality holds trivially for $(s,t)$, 
where $|M_{s,t}| \ge \Qconst/k$ or $M_{s,t} = 0$ 
since $\rho_{s,t} = 0$ in both of these edge cases.
Thus, the sum of the third central absolute moment 
of the summands we need for Berry-Esseen is
\begin{align*}
\sum_{s,t} \rho_{s,t}
\leq \;& (\sigma^2 \times \InftyBd^2 \times \Qconst) \sum_{s,t} (x_i)_s^2 (x_i)_t^2 \frac{M_{s,t}^2}{p_{s,t}^2} p_{s,t}(1-p_{s,t})\\
= \; &(\sigma^2 \times \InftyBd^2 \times \Qconst)\Var\left( (x_i x_i^\top) \bullet Q\right) \;,
\end{align*}
where the last equality is a simple calculation to calculate the term-by-term variance for $(x_i x_i^\top) \bullet Q$.
Thus, \Cref{fact:berry-esseen} implies that the Kolmogorov distance between the distribution of $(x_i x_i^\top) \bullet Q$ and the Gaussian with the same mean and variance is at most
\begin{align}
0.57 \frac{\sum_{s,t} \rho_{s,t}}{ (\Var((x_i x_i^\top) \bullet Q))^{1.5}}
&\leq 0.57 (\sigma^2 \times \InftyBd^2 \times \Qconst) \frac{\Var\left( (x_i x_i^\top) \bullet Q\right)}{\Var^{1.5}\left( (x_i x_i^\top) \bullet Q\right)}
 \le \frac{0.57 (\sigma^2 \times \InftyBd^2 \times \Qconst)}{\sqrt{V_Z}} \;,
\label{eq:BerryEsseen}
\end{align}
where the inequality comes from the assumption that the variance is at least $V_Z$.
Therefore, $(x_i x_i^\top) \bullet Q$ has at least probability 
$0.5-\frac{0.57 (\sigma^2 \times \InftyBd^2 \times \Qconst)}{\sqrt{V_Z}}$ of exceeding its expectation.
By our choice of quantities in \Cref{app:consts}, this probability is at least $0.4$.
Furthermore, $\Exp((x_i x_i^\top) \bullet Q) =x_ix_i^\top \bullet M $ 
is bigger than $\res$ and in turn bigger than $\Qscale \times \AKPtail$ (by our choice for these quantities).
Thus, with probability at least $0.4$, $(x_i x_i^\top) \bullet Q$ 
exceeds $\Qscale \times \AKPtail$.

Combined with the guarantee that 
$\Pr_{Q}((x_i x_i^\top) \bullet Q > \Qscale \times \AKPtail) > 0.99$ 
in the case where $\Var((x_i x_i^\top) \bullet Q) \le V_Z$ (cf.~\Cref{eq:resVsqscaletail}), 
we have shown that in all cases, 
$\Pr_{Q}((x_i x_i^\top) \bullet Q > \Qscale \times \AKPtail) > 0.4$ 
whenever $x_ix_i^\top \bullet M > B$.

Thus, we have shown that 
$\Exp[Z] = \sum_i \Pr_{Q}((x_i x_i^\top) \bullet Q > \Qscale \times \AKPtail) > 0.4 \times  0.01n = 0.004n$, 
since at least $0.01$ fraction of points satisfy $x_ix_i^\top \bullet M > B$ 
and thus also satisfy $\Pr_{Q}((x_i x_i^\top) \bullet Q > \Qscale \times \AKPtail) > 0.4$.
Note also that $Z \in [0,n]$ always, meaning that $\Exp[Z^2] \le n^2$.
Since $0.004 \geq 4 \times \AKPprob$ by our choice of $\AKPprob$, 
it then follows from the Paley-Zygmund inequality that
\begin{align*} 
\Pr(Z > 2\times\AKPprob\times|S|) \ge 0.25 \frac{(\Exp[Z])^2}{n^2} > \frac{0.25\times 0.004^2n^2}{n^2} = 4\cdot 10^{-6}
\end{align*}
showing the second claim above, 
and completing the proof of this lemma.
\end{proof}

\section{Smoothness of Stability Under Truncation}
\label{app:stability-smoothness}

{
The goal of this section is to prove \Cref{thm:finalStable} (restated below), 
the stability result we use in the proof of our main result, 
\Cref{thm:ourResult_Chi_k} and hence~\Cref{thm:ourResult}.}
Recall the function $h_{a,b}$ as defined in \Cref{eq:hDefn}.

\ThmFinalStable*

{In \Cref{sec:stability_after_removing_points_additive_dependence_on_}, we proved \Cref{thm:stabHighProb}.
While it is tempting to directly use \Cref{thm:stabHighProb} 
to prove the main result of \Cref{thm:ourResult_Chi_k}, 
it does not apply as-is for analyzing \Cref{alg:main}.
If we tried to use \Cref{thm:stabHighProb} to prove \Cref{thm:ourResult_Chi_k}, 
the intuitive way is to apply \Cref{thm:stabHighProb} 
to the distribution $h_{a,\mu}(X)$ for $X \sim P$ --- $h_{a, \mu}(X)$ is by construction 
bounded in $\ell_\infty$ norm, and the means and covariances 
of $P$ and $h_{a,b}(X)$ are close in $\ell_{2,k}$ norm and $\Chi_k$ norm, 
respectively, by \Cref{lem:truncationInfBd}.
However, in \Cref{alg:main}, we do not truncate samples by centering at $\mu$, 
but instead use the coordinate-wise median-of-means estimate 
as the truncation center, which is data-dependent and not any fixed vector.
Thus, we have to show that the stability result in \Cref{thm:stabHighProb} 
is insensitive to which point we center the truncation at, 
and that as a corollary, an analogous result holds even if 
we truncate using the coordinate-wise median-of-means estimate as the center.
The final statement of this section is captured by \Cref{thm:finalStable}, 
and much of this section is dedicated to showing the ``Lipschitzness'' 
of the stability of samples, as we truncate using different preliminary mean estimates.

To show this ``Lipschitzness'' property, we make the following observation.
Suppose we start with the set of $n$ i.i.d.\ samples $S$ from $P$, 
which we know by \Cref{thm:stabHighProb} contains a large subset 
$S_1$ such that $h_{a,\mu}(S_1)$ is stable with respect 
to the mean vector of the truncated distribution $\mu'= \E_{X \sim P}[h_{a,\mu}(X)]$.
Further suppose we are able to show that $S$ contains another large subset $S_2$ 
that is ``coordinate-wise regular'', meaning that for each coordinate $j \in [d]$, 
most samples in $S_2$ are close to $\mu_j$ in coordinate $j$.
Then the intersection $S_3 = S_1 \cap S_2$ enjoys both stability 
and coordinate-wise regularity, and furthermore 
the stability of $h_{a,b}(S_3)$ holds for any vector $b$ that is close to both $\mu$ and $\mu'$.
This observation is shown as \Cref{lem:LipschitzSmooth} in \Cref{sec:LipschitzSmooth}, 
and we show the existence of $S_2$ in \Cref{lem:lpexpected-smoothness} in \Cref{sec:lpexpected-smoothness}.
The proof of \Cref{thm:finalStable} combines the above two lemmas 
and \Cref{thm:stabHighProb}, and is presented in \Cref{sec:proofoffinalstable}.
}

\subsection{Lipschitzness of Truncation Under Coordinate-wise Regularity}
\label{sec:LipschitzSmooth}

As explained before this subsection, \Cref{lem:LipschitzSmooth} shows that 
if our set of uncontaminated samples $S$ is such that 
1) there exists a large  subset $S_1$ with $h_{a,\mu}(S_1)$ being stable, 
and 2) there is another large ``coordinate-wise regular'' subset $S_2$, 
then $S_3 = S_1 \cap S_2$ is a large subset of $S$ that is \emph{both} ``coordinate-wise regular'' 
and $h_{a,b}(S_3)$ is stable for any $b$ sufficiently close to $\mu$.
In the following statement, instead of saying that $h_{a,\mu}(S_1)$ and $h_{a,b}(S_3)$ are stable, 
we use the essentially equivalent condition (by \Cref{Prop:stabilityMoMCov}) 
that the $\Chi_k$-norms of the empirical covariance matrices are small.

\begin{lemma}[Lipschitzness of Truncation under Coordinate-wise Regularity]
\label{lem:LipschitzSmooth}
Let ${\mu}, \mu'$ be vectors in $\R^d$ and let $a \in  \R_+$ be greater than $2$.
Let $S=\{x_1,\dots,x_n\} \subseteq \R^d$ be the set of $n$ points.
Suppose there exist a set $S_1 \subset [n]$ satisfying the following for some $r \in \R_+$:
    \begin{align}
        |S_1| \geq 0.98 n \,\,\, \text{and}\,\,\,\left\|\frac{1}{|S_1|}\sum_{i \in S_1} (h_{a,\mu} (x_i) - \mu') (h_{a,\mu} (x_i) - \mu')^\top \right\|_{\Chi_k} \leq r \;.
    \end{align}
Suppose also that, for some $\gamma \in (0,1)$, there exist a set $S_2 \subset [n]$ satisfying the following:
    \begin{align}
    \label{eq:coordinatesubset}
        |S_2| \geq 0.99 n \,\,\, \text{and}\,\,\,\forall j\in[d]: \sum_{i \in S_2} \I_{|x_{i,j} - \mu_j | \geq a/2} \leq \gamma n \;.
    \end{align}
Then, we have the following: there exists a set $S_3 \subset [n]$ 
such that for all $b\in \R^d$  satisfying $\|b - \mu\|_\infty \leq a/2 $ and $\|b -  \mu'\|_\infty \leq a $, we have that
\begin{align}
    |S_3| \geq 0.97 n \,\,\, \text{and}\,\,\,\left\|\frac{1}{|S_3|}\sum_{i \in S_3} (h_{a,b} (x_i) - \mu') (h_{a,b} (x_i) - \mu')^\top \right\|_{\Chi_k} \leq 1.1r + 5a\gamma k\|b -\mu\|_{\infty} \;.
\end{align}
\end{lemma}
\begin{proof}
We will take $S_3 = S_1 \cap S_2$, which directly implies that $|S_3| \geq 0.97n$. 
For any $M \in \Chi_k$, since $xx^\top \bullet M \geq 0$, we have the following:
\begin{align*}
\left \langle M, \frac{1}{|S_3|}\sum_{i \in S_3} (h_{a,\mu} (x_i) - \mu') (h_{a,\mu} (x_i) - \mu')^\top \right \rangle
&\leq \left \langle M, \frac{1}{|S_3|}\sum_{i \in S_1} (h_{a,\mu} (x_i) - \mu') (h_{a,\mu} (x_i) - \mu')^\top \right   \rangle \\
&\leq \frac{1}{0.97} r. 
\end{align*}

Let $F(b)$ be the following matrix:
\begin{align*}
    F(b) = \frac{1}{|S_3|}\sum_{i \in S_3} (h_{a,b} (x_i) - \mu') (h_{a,b} (x_i) - \mu')^\top.
\end{align*}
We will establish that $\|F(b) - F(\mu)\|_{\chi_k} \leq 5a\gamma k \|b - \mu\|_{\infty}$, which establishes the lemma statement by the triangle inequality.
In order to do that, we will show that $\|F(b) - F(\mu)\|_{\infty} \leq 5a\gamma \|b - \mu\|_{\infty}$ and then use \Cref{lem:chiKSuff} below (proved in \Cref{app:chiknorm}).
\begin{restatable}{lemma}{LemChiK}
\label{lem:chiKSuff}
 Let $A \in \R^{d \times d}$ be a symmetric matrix such that $|A_{i,i}| \leq \eta_1$ for each $i \in [d]$, and $|A_{i,j}| \leq \eta_2$ for each $i \neq j \in [d] \times [d]$. Then $\|A\|_{\cX_k} \leq \eta_1 + k \eta_2$.
\end{restatable}

Consider an arbitrary $(j,\ell)$-entry of these matrices.
By abusing notation, when $x$ and $y$ are scalar, we use $h_{a,y}(x)$ to be the function from $\R \to \R$ defined analogously to \Cref{eq:hDefn}.
Let $g(\cdot,\cdot)$ be the following function that is equal to the $(j,\ell)$ entry of the matrix $F(b)$, which is explicitly
\begin{align*}
    g(b_j,b_\ell) = \frac{1}{|S_3|}\sum_{i \in S_3} (h_{a,b_j} (x_j) - \mu'_j) (h_{a,b_\ell} (x_\ell) - \mu'_\ell).
\end{align*}
We will show that $g(\cdot,\cdot)$ is locally Lipschitz in its arguments. Consider a particular $ i \in S_3$ and define the following:
\begin{align*}
    g_i(b_j,b_\ell)=  (h_{a,b_j} (x_{i,j}) - \mu'_j) (h_{a,b_\ell} (x_{i,\ell}) - \mu'_\ell).
\end{align*}

Then, we can upper bound the difference for each sample by
\begin{align*}
    |g_i(b_j,&b_\ell)  - g_i(\mu_j,\mu_\ell)|\\
    &= |(h_{a,b_j} (x_{i,j}) - \mu'_j) (h_{a,b_\ell} (x_{i,\ell}) - \mu'_\ell) - (h_{a,\mu_j} (x_{i,j}) - \mu'_j) (h_{a,\mu_\ell} (x_{i,\ell}) - \mu'_\ell)|  \\
    &\le | (h_{a,b_j} (x_{i,j}) - h_{a,\mu_j} (x_{i,j})) (h_{a,b_\ell} (x_{i,\ell}) - \mu'_\ell)| + |(h_{a,\mu_j} (x_{i,j}) - \mu'_j) (h_{a,b_\ell}(x_{i,\ell}) - h_{a,\mu_\ell} (x_{i,\ell}))| \\
    &\le (a + \|b - \mu'\|_\infty) \cdot\|b-\mu\|_{\infty} \left( \1_{|x_{i,j} - \mu_j| \ge a - \|b- \mu\|_\infty } + \1_{  |x_{i,\ell} - \mu_\ell| \ge \|b- \mu\|_\infty} \right)\\
    &\le (a + \|b - \mu'\|_\infty) \cdot\|b-\mu\|_{\infty} \left( \1_{|x_{i,j} - \mu_j| \ge a/2} + \1_{  |x_{i,\ell} - \mu_\ell| \ge a/2} \right)\\
 &\le 2a \cdot\|b-\mu\|_{\infty} \left( \1_{|x_{i,j} - \mu_j| \ge a/2} + \1_{  |x_{i,\ell} - \mu_\ell| \ge a/2} \right),
\end{align*}
where we use that $|h_{a,y}(x) - h_{a,z}(x)| \leq |y-z|$,  $|h_{a,y}(x) - z| \le |a+y-z|$, and $|h_{a,y}(x) - h_{a,z}(x)|$ is non-zero only if $|x-y| \ge a-|y-z|$.

Combined with assumption~\Cref{eq:coordinatesubset}, this implies that
\begin{align*}
    |g(b_j,b_\ell) - g(\mu_j,\mu_\ell)| \le 2a\cdot\|b-\mu\|_{\infty}\cdot\frac{2\gamma}{0.97} \leq 5a \gamma \|b-\mu\|_{\infty}.
\end{align*}

By \Cref{lem:chiKSuff}, we have the following:
\begin{align*}
    \|F(b) - F(\mu)\|_{\Chi_k} \leq k\|F(b) - F(\mu)\|_{\infty} \leq 5 a \gamma k \cdot\|b-\mu\|_{\infty}\cdot
\end{align*}

\end{proof}

\subsection{Existence of Large Subset of ``Coordinate-wise Regular'' Samples}
\label{sec:lpexpected-smoothness}

This subsection shows~\Cref{lem:lpexpected-smoothness}, 
which states that with high probability there exists a large subset of ``coordinate-wise regular'' 
samples where in each dimension at most a negligible fraction of the points have large magnitude.
As explained earlier, we will combine~\Cref{lem:lpexpected-smoothness,lem:LipschitzSmooth} 
to show~\Cref{thm:finalStable}.

For a vector $X_i \in \R^d$, we will use $X_{i,j}$ to refer to the $j$-th coordinate of $X_i$.

\begin{lemma}
\label{lem:lpexpected-smoothness}
Let $P$ be a distribution over $\R^d$, and $k \in [d]$.
For $X \sim P$, suppose for all $j \in [d]$,  $\E[X_j^4] \leq \nu^4$.
Then, there exists a positive constant $c_1$ such that, with probability at least $1-\tau$ over the set $S$ of $n \ge c_1 (k^{1.5} + \log(1/\tau))$ i.i.d.\ samples from $P$, $S$ contains a (large) subset $S'$ such that the following hold simultaneously:
\begin{enumerate}
    \item $|S'| \geq 0.99|S|$, and
    \item For each $j \in [d]$, the number of points in $S'$ with their $j$-th coordinate at least $ 2\nu \sqrt{k}$ in magnitude is at most $n/k^{1.5}$. Equivalently,
    \begin{align}
        \forall j \in [d]: \,\,\,\, \sum_{i \in S} \1_{|X_{i,j}| \geq 2 \nu \sqrt{k}} \leq  \frac{ n}{k^{1.5}} \;.
    \end{align}
\end{enumerate}
\end{lemma}
Before providing the proof of \Cref{lem:lpexpected-smoothness}, we highlight why the result is not obvious.
The first approach that one may try is to show that the original set $S$ directly satisfies the claim, that is, (with high probability) in each coordinate, the fraction of points with large magnitude in that coordinate is at most $ k^{-1.5}$.
At the population level, this is indeed true by the fourth moment assumption: for any fixed $i\in[n]$ and $j \in [d]$, the probability that $|X_{i,j}|$ is large is at most $O(1/k^2)$.
However, for this to hold with probability $1 -\tau$, one requires roughly  $k^{1.5} \log(1/\tau)$ samples \emph{even in 1 dimension}\footnote{The upper bound follows from a Chernoff bound, and the lower bound follows from the fact that Chernoff bounds are essentially tight for sums of Bernoulli coins.}, which would give a multiplicative dependence on $\log(1/\tau)$ instead of additive dependence.

The second approach that one may try would be the following: define $S'$ to be the set of all ``good'' samples, where we say a sample is ``good'' if all of its coordinates are smaller than $c \nu \sqrt{k}$.
However, for any fixed coordinate $j \in[d]$, the probability that the $j$-th coordinate of a sample being larger than $c \nu \sqrt{k}$ can be as large as $1/k^2$.
Thus, when $d \gg k^2$, the probability that a particular sample is ``bad'' may be arbitrarily close to $1$ --- for example, when coordinates are independent --- and the resulting set $S'$ will be too small with high probability.

We now give the proof of \Cref{lem:lpexpected-smoothness}, which phrases the existence of the set $S'$ as an integer program feasibility problem.
The proof considers the LP relaxation and uses LP duality techniques to show that the integer program has to be feasible.

\begin{proof}

We will assume that $k \geq C$ for a large enough constant. If $k$ is smaller than the constant, then the result follows by applying Bernstein inequality and taking $S' = S$.

Let $S = \{Y_1,\dots,Y_n\}$.
For $i \in [n]$ and $j \in [d]$, we use $Z_{i,j}$ to denote $\I_{|Y_{i,j}| \geq c_2 \nu \sqrt{k}}$.
For simplicity, we set $\alpha =  k^{-1.5}/3$.
Our goal is to show that the following integer program is feasible:
\begin{align}
\label{eq:primal-integer}
\tag*{(F1)}
   \begin{split}
\textrm{variables} \quad & p_1,\dots,p_n\\
\textrm{subject to} \quad & \forall j \in [d]: \sum_{i=1}^n p_i Z_{i,j} \leq \alphacoordinates n  \\
\quad & \sum_{i=1}^n p_i \geq
0.99 n\\
 \quad & \forall i \in [n]: p_i \in \{0,1\}.
\end{split}
\end{align}
As argued above in the prose after the statement, one needs to argue about all the samples, and their coordinates,  simultaneously to prove the statement.
Since directly handling the feasibility program \ref{eq:primal-integer} seems difficult, our argument will go in the following steps: (i) first consider the LP relaxation of \ref{eq:primal-integer}, (ii) using duality theory, the LP relaxation is feasible if and only if the dual LP is infeasible, (iii) simplify the dual LP and show that, with high probability, the resulting program is infeasible.

We begin by considering the LP relaxation.
\begin{align}
\label{eq:primal-lp}
\tag*{(F2)}
\begin{split}
\textrm{variables} \quad & p_1,\dots,p_n\\
\textrm{subject to} \quad & \forall j \in [d]: \sum_{i=1}^n p_i Z_{i,j} \leq \alphacoordinateslp n  \\
\quad & \sum_{i=1}^n p_i \geq
 0.999 n\\
 \quad& \forall i \in [n]: p_i \in [0,1].
\end{split}
\end{align}
We first show that if the above LP relaxation, \ref{eq:primal-lp}, is feasible, then \ref{eq:primal-integer} is also feasible.
\begin{restatable}[Feasibility of \ref{eq:primal-lp} implies feasibility of \ref{eq:primal-integer}]{claim}{ClaimLPToInteger}
\label{claim:LPToInteger}
Suppose $n > 10^{6}$ and $\alpha \geq (4 \log n)/n$. If \ref{eq:primal-lp} is feasible, then \ref{eq:primal-integer} is also feasible.
\end{restatable}
\begin{proof}
Let $p_1,\dots,p_n$ be the feasible solution to \ref{eq:primal-lp}. Consider the following random assignment, for $i \in [n]$, $P_i \sim \Bernoulli(p_i)$ independently.
We will show that, with non-zero probability, $P_i$'s satisfy \ref{eq:primal-integer}.
We will use the following inequality:
\begin{fact}[Chernoff Inequality] 
\label{fact:chernoff}
Let $a_1,\dots,a_n$ such that $a_i \in \{0,1\}$. Let $W_1,\dots,W_n$ be independent Bernoulli random variables and consider the random variable $Z = \sum_{i=1}^n W_i$. Then, with probability $1 -\tau$, $Z \leq 2 (\E Z + \log(1/\tau))$.
\end{fact}
By \Cref{fact:chernoff}, we get that each of the inequalities in \ref{eq:primal-integer} holds with probability $1- 1/(2n)$ as long as $n \alpha \geq  2\log (2n)$ and $n > 1000 \log (2n)$. The latter holds when $n \geq 10^6$.
\end{proof}

Since $n> 10^6$ in our setting (as $k$ is large and choosing $c_1$ to be large enough) and $\alpha = 1/(3k^{1.5})$, we have that $\alpha \geq 4 (\log n)/n $ is equivalent to $n \geq 12 k^{1.5} \log n $, which is satisfied when $n \geq 100  k^{1.5} \log k$. The latter holds when $n \geq c k^{1.5} \log d$ for a large enough constant $c$.
Thus, in the remainder of this section, we will show that, with high probability, this LP program is indeed feasible.
We begin by considering the following dual program:
\begin{align}
\label{eq:dual-lp}
\tag*{(F3)}
\begin{split}
\textrm{variables} \quad & w_1,\dots,w_d, y_1,\dots,y_n, x\\
\textrm{subject to} \quad & \sum_{i=1}^n y_i + \alphacoordinateslp n \sum_{j=1}^d w_j 
< 0.999 nx \\
\quad & 
 \forall i \in [n]: y_i + \sum_{j=1}^d Z_{i,j} w_j \geq x \\
 \quad& z \geq 0, \quad\forall i \in [n]: y_i \geq 0,\quad \forall j \in [d]: w_j \geq 0.
 \end{split}
\end{align}
Suppose for the sake of contradiction that \ref{eq:primal-lp} is infeasible.
By Farkas' lemma~\cite{gale1951linear}, it means that the (dual) program in \ref{eq:dual-lp} is feasible.
Formally, we have the following claim:
\begin{restatable}[LP Duality for \ref{eq:primal-lp}]{claim}{CLaimLPDuality}
\label{claim:lp-duality}
\ref{eq:dual-lp} is infeasible if and only if \ref{eq:primal-lp} is feasible.
\end{restatable}
\Cref{claim:lp-duality} follows from Farkas' lemma.
We will argue that \ref{eq:dual-lp} is infeasible by showing that the following program, which is feasible whenever \ref{eq:dual-lp} is feasible,  is infeasible.
\begin{align}
\label{eq:dual-lp2}
\tag*{(F4)}
\begin{split}
\textrm{variables} \quad & w_1,\dots,w_d, A \\
\textrm{subject to} 
\quad & 
 \forall i \in A: \sum_{j=1}^d Z_{i,j} w_j \geq \alphacoordinateslp (\sum_{j=1}^d w_j)  \\
 \quad&  \forall j \in [d]: w_j \geq 0, \\
 \quad& A \subset [n], |A| \geq
10^{-3} n
 \end{split}
\end{align}
\ref{eq:dual-lp2} states that for at least $10^{-3}$ fraction of $i$'s in $n$, the following inequality holds: $\sum_{j=1}^d Z_{i,j} w_j \geq \alpha \|w\|_1$.
The following claim relates the two programs above.
\begin{restatable}
{claim}{ClaimF3ToF4}
If \ref{eq:dual-lp} is feasible, then \ref{eq:dual-lp2} is feasible. 
\end{restatable}
\begin{proof}
Let $y_1,\dots,Y_n, w_1,\dots,w_d,x$ be any feasible solution to \ref{eq:dual-lp}.
Then the first constraint in \ref{eq:dual-lp} that the average of $y_i$'s is less than $0.999 x - \alphacoordinateslp (\sum_{j=1}^d w_j)$.
By Markov's inequality,  the fraction of the $y_i$'s such $y_i \geq ( x - \alphacoordinateslp (\sum_{j=1}^d w_j))$ is at most $ \frac{0.999 x - \alphacoordinateslp (\sum_{j=1}^d w_j)}{( x - \alphacoordinateslp (\sum_{j=1}^d w_j))} \leq 0.999$.
Thus the fraction of $y_i$'s such that $y_i < ( x - \alphacoordinateslp (\sum_{j=1}^d w_j))$ is at least $0.001$.

Let $A \subset [n]$ be the set of such indices. 
For any $i \in A$, the second  constraint in \ref{eq:dual-lp} implies that $\sum_{j=1}^d Z_{i,j} w_j \geq x - y_i \geq \alphacoordinateslp (\sum_{j=1}^d w_j)$.
This implies that \ref{eq:dual-lp2} is feasible.
\end{proof}
In order to argue that \ref{eq:dual-lp2} is infeasible, we first consider a particular $w$. Using calculations provided below, it can be seen that the probability that a particular $w$ satisfies \ref{eq:dual-lp2} is exponentially small in $n$. However, a direct approach at covering $w$ seems difficult since $w$ is a dense vector in $\R^d$ and $n = o(d)$. 
Using a randomized rounding mechanism, we show that it suffices to consider only sparse $w$ as follows:
\begin{align}
\label{eq:dual-lp3}
\tag*{(F5)}
\begin{split}
\textrm{variables} \quad & w_1,\dots,w_d, A \\
\textrm{subject to} 
\quad & 
 \forall i \in A: \sum_{j=1}^d Z_{i,j} w_j \geq  1  \\
 \quad&  \forall j \in [d]: w_j \in \{0,1\}, \\
 \quad& \sum_{j=1}^d w_j \leq
    \frac{2 \times 10^{7}}{\alpha} \\
 \quad& A \subset [n], |A| \geq 
 10^{-4} n
 \end{split}
\end{align}
The following claim shows that if \ref{eq:dual-lp2} is feasible then \ref{eq:dual-lp3} is also feasible.
\begin{restatable}{claim}{ClaimSparseW} If \ref{eq:dual-lp2} is feasible, then \ref{eq:dual-lp3} is also feasible.
\end{restatable}
\begin{proof}
Let $w_1,\dots,w_d$ and $A$ be the feasible solution to \ref{eq:dual-lp3}. 
Set $q_j = \min(1,   w_j/ (\alpha\|w\|_1) )$ for $j \in [d]$.
Consider the following random assignment: set $W_j \sim \Bernoulli(q_j)$ independently for $j \in [d]$.
We will show that with non-zero probability $W_j$'s satisfy \ref{eq:dual-lp3}.
Consider the following events:
\begin{align}
    \cE_1 := \left\{\sum_{j=1}^d W_j \leq 2 \times 10^{-7}(1/\alpha)\right\}, \,\,\,\, \text{and}  \,\,\,\, \cE_2 := \left\{|\{i: \sum_{j=1}^d Z_{i,j}W_j \geq 1\}| \geq 10^{-4} n\right\} 
\end{align}
We will show that $\P\{cE_1\} \geq 1- 5 \times 10^{-8} $ and $\cP\{cE_2\} \geq 10^{-7}$.
By a union bound, we will have that $\cE_1 \cap \cE_2$ has non-zero probability and thus \ref{eq:dual-lp3} is feasible.

Let $F := \{j\in[d]: W_j = 1\}$ be the set of coordinates where $W_j$ is non-zero. 
Then $\E[|F|] = \E[\sum_{j=1}^d W_j] = \sum_{j=1}^d q_j \leq 1/\alpha$.
Thus with probability at least $1- 5 \times 10^{-8}$, we have that the number of non-zero $W_j$'s is at most  $ \frac{2 \times 10^{7}}{\alpha} $.
Equivalently, $\P\{\cE_1\} \geq 1- 5 \times 10^{-8}$.

We now focus on the second event $\cE_2$. Let $S_1,\dots,S_n$ be the subsets of $[d]$ such that $S_i = \{j \in [d]: Z_{i,j} = 1\}$, i.e., for each sample $i$, $S_i$ is the set of indices where the coordinates are large.
Consider the random variables $R_1,\dots,R_n$, where for $i \in [n]$, $R_i:= \sum_{j=1}^d Z_{i,j} W_j = \sum_{j \in S_i} Z_{i,j}W_j$. \ref{eq:dual-lp3} requires that for at least $10^{-4}$ fraction of $i$'s, $R_i \geq 1$. 
Since $Z_{i,j}$'s are binary and fixed, we have that $R_i$ is distributed as Binomial random variable and is thus anti-concentrated.
\begin{fact}[Anti-concentration of Binomial] 
\label{fact:ac-binom}
Let $X \sim \text{Binomial}(n,p)$ for some $n \in \N$ and $p\in [0,1]$. Suppose $\E[X] \geq 1$.
Then $\P\{X \geq 1\} \geq  (1 - 1/e)$.
\end{fact}
\begin{proof}
Using the fact that $1 + x \leq e^{x}$ for all $x \in \R$, we get the following:
\begin{align*}
    \P\{X \geq 1\} &= 1 - \P\{X = 0 \} = 1 - (1 - p)^n\geq 1- (e^{-p})^n = 1 - e^{-np} \geq (1 - 1/e).\qquad \qquad \qedhere
\end{align*}
\end{proof}

Consider a fixed $i \in A$.  Then either there exists a $j \in S_i$ such that $q_j =1$, or for all $j \in S_i$,  $q_j < 1$.
In the former case, we have that $R_i$ is at least one since $W_j =1 $.

In the latter setting, we have that $q_j = w_j/(\alpha \|w\|_1)$ for all $j\in S_i$, and thus $\E[R_i] = \sum_{j \in S_i}^d Z_{i,j} q_j = \sum_{j=1}^d Z_{i,j} w_j/ (\alpha \|w\|_1) \geq 1$.
 Applying \Cref{fact:ac-binom} to any such $i \in A$, we get that the probability of $R_i$ being positive is at least $1 - 1/e$.
Let $A'$ be set of $i$'s such that $R_i \geq 1$, i.e., $A' = \{i: R_i \geq 1\}$. Thus combining the two cases above, we have the following:
\begin{align}
\label{eq:sparsifyingLP}
    \forall i \in A: \P\{i \in A'\} \geq 0.5.
\end{align}
Thus $\E[|A'|] \geq 0.5 |A| \geq 5 \times 10^{-4}$. Since $|A'|$ lies in $[0,n]$, applying Paley-Zygmund inequality to the random variable $|A'|$, we get the following:
\begin{align}
\P\{|A'| \geq 10^{-4}  n\} \geq  \P\{|A'| \geq 0.2 \E[|A'|]\} \geq 0.64 \frac{(\E[|A'|])^2}{n^2} \geq 0.64 \times 25 \times  10^{-8} > 10^{-7}. 
\end{align}
Equivalently, $\P\{\cE_2\} \geq 10^{-7}$.
This completes the proof.
\end{proof}

Thus, it suffices to show that, with high probability, \ref{eq:dual-lp3} is infeasible.
\begin{lemma}[Infeasibility of \ref{eq:dual-lp3}]
\label{lem:infeasible-sparse}
Under the setting of \Cref{lem:lpexpected-smoothness} and when $k> 10^{26}$, there exists a constant $c_1 > 0$ such that if $n \geq c_1(k^{1.5}\log d + \log(1/\tau))$, then with probability $1 - \tau$, \ref{eq:dual-lp3} is infeasible.
\end{lemma}
\begin{proof}
First consider any fixed $w= (w_1,\dots,w_d)$ such that $w_i \in \{0,1\}$ and $\sum_{j=1}^d w_j \leq 2 \times 10^{7} \cdot (1/\alpha)$. 

Consider the integer-valued random variables $R_1,\dots,R_n$ such that $R_i = \sum_{j=1}^d Z_{i,j} w_j$, and observe that $R_i$'s are i.i.d. random variables (since $X_i$'s are i.i.d.\ random variables).
Thus, \ref{eq:dual-lp3} requires that at least $10^{-4}\%$ of $R_i$'s are non-zero.

By the fourth moment bound on each coordinate, we have that $\E[Z_{i,j}] = \P\{X_{i,j} \geq 2 \nu \sqrt{k}\} \leq 1/k^2$ for each $i$ and $j$.
Therefore, the expectation of each $R_i$ is at most $\sum_{j=1}^d w_j\E[Z_{i,j} ] \leq \sum_{j=1}^d w_j (1/k^2) \leq (2 \times 10^{7})/ (k^2 \alpha) =  (2 \times 10^{7})/ (k^2 \alpha)  = (6 \times 10^{7})/\sqrt{k}$, which is less than $10^{-5} $ for $k$ large enough.
By Markov's inequality, the probability that $\P\{R_1 \geq 1\} \leq 10^{-5}$.

Hence, by the Chernoff bound (since $R_i$'s are independent), with probability at least $1 - \exp(-c' n)$, the fraction of $R_i$'s that are non-zero is at most $5\times 10^{-5}$.
Hence, with the same probability, this particular choice of $w$ does not satisfy \ref{eq:dual-lp3}.
Since there are at most $d^{ (2 \times 10^{7})/\alpha)}$ such choices of $w$, applying a union bound, we get that \ref{eq:dual-lp3} is infeasible with probability at least $1 - \exp( (2 \times 10^{7})/\alpha) \cdot \log d - c'n)$. The failure probability is at most $\tau$ when $n \gtrsim \log (1/\tau) +  k^{1.5}\log d$.
This concludes the proof.%
\end{proof}
Since we assumed $k$ is large enough, \Cref{lem:infeasible-sparse} is applicable. 
\Cref{lem:infeasible-sparse} implies that, with high probability, the program \ref{eq:dual-lp3} is infeasible.
Hence, with the same high probability, the programs \ref{eq:dual-lp2} and \ref{eq:dual-lp} are also infeasible, and the programs \ref{eq:primal-integer} and \ref{eq:primal-lp} are feasible. 
This completes the proof.
\end{proof}

\subsection{Proof of \Cref{thm:finalStable}}
\label{sec:proofoffinalstable}

We now combine \Cref{lem:lpexpected-smoothness,lem:LipschitzSmooth} to show~\Cref{lem:finalStable}, stating that with high probability over the uncontaminated and untruncated samples $S$, there is a large subset $S'$ such that for any truncation center $b$ that is close to the true mean, $h_{a,b}(S')$ is stable for $a = \Theta(\sigma \sqrt{k})$ as chosen in \Cref{alg:main}.
\Cref{thm:finalStable} follows a corollary, by instantiating $b$ to be the coordinate-wise median-of-means estimate.

\begin{lemma}
\label{lem:finalStable}
Let $S$ be a set of $n$ i.i.d.\ data points from a distribution $P$ over $\R^d$. Let the mean of $P$ be  $\mu$, and covariance $\Sigma$ such that $\|\Sigma\|_{\Chi_k} \le \sigma^2$, and for all $i \in [d]$, $\E[X_i^4] \le O(\sigma^4)$.
Suppose $n = \Omega(k^2 \log d + \log(1/\tau))$. Let $a = 4\sigma\sqrt{k}$.
With probability $1- \tau$ over $S$, there exists a subset $S' \subset S$ with $|S'| \ge 0.95 n$ such that for any $b$ satisfying $\|b - \mu \|_\infty = O (\sigma) $, we have $h_{a,b}(S')$ is $(0.01, O(1) ,k)$-stable with respect to some $\mu'$ and $\sigma$ with $\|\mu'-\mu\|_{\infty} \le O(\sigma/\sqrt{k})$.
\end{lemma}

\begin{proof}

Let $P'$ be the distribution of $h_{a,\mu}(P)$ and let $\mu'$ and $\Sigma'$ be the mean and covariance of $P'$.
This will be the $\mu'$ in the lemma statement.
By \Cref{lem:truncationInfBd}, we get that (i) $\|\mu' - \mu\|_{\infty} \leq \sigma /\sqrt{k}$, (ii) $\|\Sigma -\Sigma'\|_{\Chi_k} \leq O(\sigma^2 )$, (iii) $P'$ is supported on the set $\{x: \|x - \mu'\|_{\infty} \leq 2a\} $,
 and (iv) the axis-wise fourth moment of $P'$ is upper bounded by a constant multiple of that of $P$. Thus \Cref{thm:stabHighProb} can be applied to $P'$.

Applying \Cref{thm:stabHighProb} to $P'$ gives that, with probability at least $1-\tau$, there exists a subset $S_1 \subset S$ (which are samples from $P$, before truncation) with $|S_1| \ge 0.98n$ such that $h_{a,\mu} (S_1)$ is $(0.01, O(1), k)$-stable with respect to $\mu'$.
In particular, we have
\begin{align}
        \left\|\frac{1}{|S_1|}\sum_{i \in S_1} (h_{a,\mu} (x_i) - \mu') (h_{a,\mu} (x_i) - \mu')^\top \right\|_{\Chi_k} \leq O(\sigma^2).
\end{align}

Let $r:= \sigma/(\max_{j \in [d]}\Exp_{X \sim P}[X_j^4])^{1/4}$.
By our assumption on $P$ in the lemma, we have $r = \Theta(1)$.
By applying \Cref{lem:lpexpected-smoothness} to $P - \mu$, and using $kr^2$ in place of $k$ in \Cref{lem:lpexpected-smoothness}, with probability at least $1-\tau$, there exists a subset $S_2 \subset S$ with $|S_2| \ge 0.99n$ such that for all $j\in[d]$,
    \begin{align}
     \sum_{i \in S_2} \I_{|x_{i,j} - \mu_j | \geq a/2} = \sum_{i \in S_2} \I_{|x_{i,j} - \mu_j | \geq 2 \nu \sqrt{k r^2}} \leq   O(k^{-1.5}r^{-3}) n \leq O(k^{-1.5})n.
    \end{align}

We can then apply \Cref{lem:LipschitzSmooth} to show that, conditioned on the above two existence events, there exists a third subset $S_3 \subset S$ with $|S_3| \ge 0.97n$ such that
for all $b$ satisfying $\|b - \mu\|_\infty \leq O(\sigma) $ and $\|b - \mu'\|_\infty \leq O(\sigma)$ (the latter holds by the triangle inequality for all $b$ with $\|b - \mu\|_{\infty} \leq O(\sigma)$), we have that
\begin{align}
    \left\|\frac{1}{|S_3|}\sum_{i \in S_3} (h_{a,b} (x_i) - \mu') (h_{a,b} (x_i) - \mu')^\top \right\|_{\Chi_k} \leq O(\sigma^2) + O(a k^{-1.5} k \|b-\mu\|_{\infty}) \le O(\sigma^2).
\end{align}

By \Cref{Prop:stabilityMoMCov}, this implies $S_3$ contains a set $S'$ satisfying the following: (i) $|S'| \geq 0.95n $ and (ii) $h_{a,b}(S')$ is $(0.1, O(1),k)$-stable with respect to $\mu'$ and $\sigma$, for any $b$ satisfying $\|b-\mu\|_\infty \le O(\sigma)$.
Thus, we choose $S'$ in the lemma statement to be this set.

Taking a union bound, all the above events fail with probability at most $O(\tau)$.
Reparameterizing yields the lemma statement.

\end{proof}
\ThmFinalStable*
\begin{proof}
By \Cref{fact:co-MoM}, we know that with probability at least $1-\tau$, we have $\|\tilde{\mu}-\mu\|_{\infty} \le O(\sigma)O(1+(\log(d/\tau))/n) = O(\sigma)$ by the assumption that $n$ is sufficiently large.

Thus, we use $\tilde{\mu}$ as ``$b$'' in \Cref{lem:finalStable} to yield the stability guarantee in the theorem statement.

The total failure probability is at most $2\tau$, and reparameterizing yields the theorem statement.
\end{proof}

\section{Information-Theoretic Lower Bound} %
\label{sec:information_theoretic_error}
In this section, we show that the asymptotic error of \Cref{thm:ourResult} is optimal under a mild assumption on $k$.
Let $\cD_{k}$ be the family of all distributions over $\R^d$ that satisfy the following:
\begin{enumerate}
   \item For every $D \in \cD_{k}$, the mean of $D$ is $k$-sparse,
   \item For every $D$ in $\cD_{k}$ the covariance of $\cD$ is upper bounded by $I$ in spectral norm, and
   \item For every $D \in \cD_{k}$ we have that $\E[(X_i - \E[X_i])^4] = O(1)$, where $X = (X_1,\dots,X_d) \sim D$. 
 \end{enumerate} 

\begin{restatable}{lemma}{LemInfoThError}
\label{lem:InfoTheoretic4thMoment}
Let $k \geq 1/ \sqrt{\epsilon}$. Then there exist two distributions in $D_1$ and $D_2$ in $\cD_{k}$ such that the following hold:
  (i) $d_{\text{TV}}(D_1,D_2) = \epsilon$, and
  (ii)
  The means of $D_1$ and $D_2$ are separated by $\Omega(\sqrt{\epsilon})$ in $\ell_{2,k}$-norm.
\end{restatable}

Before giving the proof of \Cref{lem:InfoTheoretic4thMoment}, we remark that the assumption of $k \ge 1/\sqrt{\eps}$ is mild.
First, the assumption is independent of the ambient dimensionality $d$---the most challenging parameter regime in algorithmic robust statistics is when we fix a small $\eps$ and then take the dimensionality $d$ to $\infty$.
Second, the typical interesting sparsity regime is when $k$ is super-constant but grows very slowly in $d$, say, logarithmically.
The assumption that $k \ge 1/\sqrt{\eps}$ applies readily to the above regime.

\begin{proof}
Let $D_1$ be the distribution that places all of its mass at origin, i.e., $(0,\dots,0)$.
Let $D_2$ be the distribution that places $(1 - \epsilon)$ probability mass at origin and places  $\epsilon$ probability mass at $y$, where the first $k$-coordinates of $y$ are $\alpha$ for some $\alpha$ to be decided later and the remaining $d-k$ coordinates are $0$.

It is easy to see that the total variation distance between $D_1$ and $D_2$ is $\epsilon$, and that $D_1 \in \cD_{k}$. We will now show that $D_2 \in  \cD_{k}$ for a suitable value of $\alpha$.
\begin{enumerate}
	\item First the mean of $D_2$ is $\epsilon y$, which is $k$-sparse by construction.

\item We have that the covariance of $D_2$ is $ \epsilon yy^\top  - \epsilon^2 yy^\top  = \epsilon( 1 -\epsilon) yy^\top  \preceq \epsilon yy^\top $, which is upper bounded by $1$ in spectral norm if $\|y\|_2 \leq 1/\sqrt{\epsilon}$.
Since $\|y\|_2 = \sqrt{k} \alpha$, we want that $\alpha \leq 1/\sqrt{k \epsilon}$.

\item Finally, let $X \sim D_2$. For every $i > k$, we have that  $\E[(X_i - \E[X_i])^4] = 0$.
For $i \in [k]$, $\E[(X_i - \E[X_i])^4] = \E[(X_i - \epsilon \alpha)^4] \leq 8 (\E[X_i^4 + \epsilon^4 \alpha^4 ]) = 8(\epsilon \alpha^4 + \epsilon^4 \alpha^4) \leq  16\epsilon \alpha^4$, which is less than $16$, if $\alpha \leq \epsilon^{-1/4}$.
\end{enumerate}
Thus, the above construction goes through as long as $\alpha \leq \min (1/\sqrt{k \epsilon}, \epsilon^{-1/4})$.
When $k \geq 1/ \sqrt{\epsilon}$, it suffices that $\alpha = 1/\sqrt{k \epsilon}$.
Finally, we note that the difference in means of $D_1$ and $D_2$ is $\epsilon \|y\|_2 = \epsilon \sqrt{k} \alpha = \sqrt{\epsilon}$ for the chosen value of $\alpha$.
\end{proof}

\bibliographystyle{alpha}

\bibliography{allrefs}

\newpage
\appendix

\section*{Supplementary Material}

\paragraph{Additional Notation}
For a vector $x \in \R^d$ and $H \subset [d]$, we denote $v_H$ to denote the vector that is equal to $v$ on $i \in H$, and zero otherwise. For a real-valued random variable $X$ and $m \in \N$, we use $\|X\|_{L_m}$ to denote $(\E|X|^m)^{1/m}$.

\section{Miscellaneous Lemmas and Facts}
\label{app:misc}

\subsection{Finding a Stable Subset from a Stable Weighted Subset}
\label{app:A1}
For a set $S$ on $n$ points, we define $\Delta_{n,\eps}$ as the set of weights $w \in \Real^n$ such that $w_i \in[0, 1/((1-\eps)n]$ for all $i \in [n]$ and $\sum_i w_i = 1$.
\new{For a fixed vector $\mu \in \R^d$ that will be clear from context, a set of $n$ points $S = \{x_1,\dots,x_n\}$, and weights $w \in \Delta_{n,\epsilon}$ over $S$, we use $\overline{\Sigma}_w$ to denote $\sum_iw_i(x_i-\mu)(x_i-\mu)^\top$.
}

The goal of this section is to show \Cref{Prop:stabilityMoMCov}, which states that if we have a weight $w$ over $S$ such that $\overline{\Sigma}_w$ (with respect to some vector $\mu$) has bounded $\Chi_k$ norm proportional to $\sigma^2$ for some $\sigma > 0$, then there must exist some large subset $S' \subseteq S$ that is stable with respect to $\mu$ and $\sigma$.

\begin{proposition}
\label{Prop:stabilityMoMCov}
Let $S$ be a set of $n$ points in $\R^d$.
Let $\Delta_{n,\eps}$ be the set of weights defined above, and define the notation $\overline{\Sigma}_w = \sum_{x_i \in S} w_i(x_i -\mu)(x_i-\mu)^\top$ for some given vector $\mu \in \R^d$.
Suppose that there exists a $w \in \Delta_{n,\eps}$ such that $\|\overline{\Sigma}_w\|_{\Chi_k} \leq B \sigma^2$ for some vector $\mu$.
Then there exists a subset $S' \subseteq S$ such that $(i) |S'| \geq (1 - 2 \eps)n$ and (ii) $S'$ is $(\eps, \delta,k)$-stable with respect to $\mu$ and $\sigma$, where $\delta = O(\sqrt{B} + 1)$.
\end{proposition}

\new{Observe that $\|\overline{\Sigma_w}\|_{\Chi_k} \le B \sigma^2$ implies $\|\overline{\Sigma_w} - \sigma^2 I \|_{\Chi_k} \le (B+1) \sigma^2$ by the triangle inequality.
In order to show \Cref{Prop:stabilityMoMCov}, we show \Cref{lemma:weightRounding}, which is a weakening of \Cref{Prop:stabilityMoMCov} where we additionally assume that $\mu_w = \sum_i w_i x_i$ is close to $\mu$, where $\mu$ is the vector we use to define $\overline{\Sigma_w}$ as well as the vector that we want to find a large sample subset $S'$ to be stable with respect to.
To use \Cref{lemma:weightRounding}, we additionally show \Cref{prop:bddcovstability}, which states that $\|\overline{\Sigma_w}\|_{\Chi_k} \le B \sigma^2$ is enough to imply that $\mu_w$ is close to $\mu$.
We combine \Cref{lemma:weightRounding} and \Cref{prop:bddcovstability} to prove \Cref{Prop:stabilityMoMCov} at the end of \Cref{app:A1}.
}

\begin{lemma}
\label{lemma:weightRounding}
Suppose, for some $\eps \le \frac{1}{3}$ and for some $\delta \ge \sqrt{\eps}$, there \new{exist a $w \in \Delta_{n,\eps}$  over a set of $n$ samples $S = \{x_1, \ldots, x_n\}$, a $\mu \in \R^d$ and a $\sigma > 0$} such that
\begin{itemize}
    \item $\|\mu_w - \mu\|_{2,k} \le \delta \sigma$,
    \item $\new{\|\sum_{i \in [n]} w_i(x_i -\mu)(x_i-\mu)^\top - \sigma^2 I\|_{\Chi_k}} \le \new{\sigma^2} \frac{\delta^2}{\eps}$.
\end{itemize}
Then, there exists a subset $S' \subseteq S$ of samples such that
\begin{itemize}
    \item $|S'| \ge (1-2\eps)|S|$,
    \item $S'$ is $(\eps,\delta',k)$-stable with respect to $\mu$ and $\sigma$, where $\delta' = O(\delta + \sqrt{\eps})$.
\end{itemize}
\end{lemma}

\begin{proof}
\new{Without loss of generality, we will only handle the $\sigma = 1$ case to simplify notation.}

The main step is to show the existence of a large subset $S'$ whose mean is within $10\delta+10\sqrt{\eps}$ of $\mu$ and whose variance is at most $9(1+ \delta^2/\eps)$.
In fact, we can simply choose $S'$ to be the subset whose weights $w_i$ are the largest.

Without loss of generality, assume $\mu = 0$ and that $\eps n$ is an integer.
We also order the samples in decreasing order of weight in $w$, namely, $1/((1-\eps)n) \ge w_1 \ge w_2 \ge \ldots \ge w_n$.

First, we will lower bound each $w_i$.
We have that for each $k \in [n]$,
\[ 1 = \sum_i w_i \le \frac{k}{(1-\eps)n} + (n-k) w_k, \]
which upon rearranging implies that
\[ w_k \ge \frac{(1-\eps)n - k}{(1-\eps)n(n-k)}. \]
In particular, for $k = (1-2\eps) n$, we have
\[ w_{(1-2 \eps) n} \ge \frac{1}{2(1-\eps)n}. \]

Letting $S'$ to be the $(1-2\eps)n$ points with the largest weight, we have that for all $i \in S'$, $w_i \ge \frac{1}{2(1-\eps)n}$.
We will use this to now bound the $\Chi_k$ norm of $\Sigma_{S'} = \frac{1}{|S'|}\sum_{i\in S'} x_i x_i^\top$.
Consider an arbitrary $M \in \Chi_k$, we have
\begin{align*}
    \sum_{i \in S'} \frac{1}{|S'|} \langle x_i x_i^\top, M \rangle &= \sum_{i \in S'} \frac{1}{(1-2\eps)n} \langle x_i x_i^\top, M \rangle\\
    &\le \sum_{i \in S'} \frac{2(1-\eps)}{1-2\eps} w_i \langle x_i x_i^\top, M \rangle\\
    &\le \sum_{i \in S} \frac{2(1-\eps)}{1-2\eps} w_i \langle x_i x_i^\top, M \rangle\\
    &\le 9 \left(1+\frac{\delta^2}{\eps}\right).
\end{align*}

Since $\delta \ge \sqrt{\eps}$, this in turn implies the (rather loose in constants) inequality that $\|\Sigma_{S'}-I\|_{\Chi_k} \le 20(\delta^2/\eps)$.

Next, we show that the mean $\mu_{S'}$ of $S'$ is $10\delta+10\sqrt{\eps}$-close to $\mu = 0$.
This will essentially follow from 1) the uniform distribution $U_{S'}$ over $S'$ is close in total variation distance to $w$ and 2) the contribution of the tail to the mean of a bounded-covariance distribution is small.

For 1), using the notation that $U_{S}$ is the uniform distribution over $S$ (analogous to the $S'$ notation just before), it is immediate that by the triangle inequality,
\[\DTV(w,U_{S'}) \le \DTV(w, U_S) + \DTV(U_S, U_{S'}) \le \eps + 2\eps = 3\eps.\]

A standard consequence is that there exists distributions $p^{(1)}$, $p^{(2)}$ and $p^{(3)}$ such that
\[ w = (1-3\eps)p^{(1)} + 3\eps p^{(2)} \quad \text{and} \quad U_{S'} = (1-3\eps)p^{(1)} + 3\eps p^{(3)}. \]

Intuitively, treating $p^{(2)}$ and $p^{(3)}$ as the ``tails'', we will bound their contributions to the mean under the boundedness of the covariance of $w$ and $U_{S'}$.

Take any $k$-sparse unit vector direction $v \in \cU_k$, we can bound the following variances in the direction of $v$:
\[ 3\eps \sum_i p_i^{(2)} \langle x_i, v \rangle^2 \le \sum_i w_i \langle x_i, v \rangle^2 \le 1+\frac{\delta^2}{\eps}, \]
\[ 3\eps \sum_i p_i^{(3)} \langle x_i, v \rangle^2 \le \sum_i U_{S',i} \langle x_i, v \rangle^2 \le 9\left(1+\frac{\delta^2}{\eps}\right), \]
where we used the fact that $vv^\top$ is in $\Chi_k$ for a $k$-sparse unit vector $v$.

By Jensen's inequality, we can then conclude that
\[ \left|3\eps \sum_i p_i^{(2)} \langle x_i, v \rangle\right| \le \sqrt{3\eps}\sqrt{3\eps \sum_i p_i^{(2)} \langle x_i, v \rangle^2} \le \sqrt{3\eps}\sqrt{1+\frac{\delta^2}{\eps}} \le \sqrt{3}(\sqrt{\eps} + \delta),\]
\[ \left|3\eps \sum_i p_i^{(3)} \langle x_i, v \rangle\right| \le \sqrt{3\eps}\sqrt{3\eps \sum_i p_i^{(3)} \langle x_i, v \rangle^2} \le 3\sqrt{3\eps}\sqrt{1+\frac{\delta^2}{\eps}} \le 3\sqrt{3}(\sqrt{\eps}+\delta).\]

Finally, since $U_{S'} = w - 3\eps p^{(2)} + 3\eps p^{(3)}$, by the triangle inequality, we have
\begin{align*}
    \left|\langle \mu_{S'} - \mu, v \rangle\right| &= \left|\sum_i U_{S',i} \langle x_i, v \rangle \right|\\
    &\le \left|\sum_i w_i \langle x_i, v \rangle \right| + \left|3\eps \sum_i p_i^{(2)}\langle x_i, v \rangle\right| + \left|3\eps \sum_i p_i^{(3)} \langle x_i, v \rangle\right|\\
    &\le \delta + \sqrt{3}(\sqrt{\eps} + \delta) + 3\sqrt{3}(\sqrt{\eps}+\delta)\\
    &\le 10\delta + 10\sqrt{\eps},
\end{align*}
where the second inequality uses the above bounds as well as the assumption that $\|\mu_w - \mu\|_{2,k} \le \delta$.

Now that we have shown that $\mu_{S'}$ is close to $\mu$ in $\ell_{2,k}$ norm and $\Sigma_{S'}$ is small in the $\Chi_k$ norm, we will use the following lemma (\Cref{lem:GoodMeanGoodVarianceStable}) to show that the set $S'$ is $(\eps,O(\delta+\sqrt{\eps})$-stable with respect to $\mu$.
\end{proof}

\begin{lemma}[Bounded Mean and Covariance implies $O(\sqrt{\eps})$ stability]
\label{lem:GoodMeanGoodVarianceStable}
\new{Let  $\mu \in \R^d$} and let $S'$ be a set of samples such that $\|\mu_{S'}-\mu\|_{2,k} \le \delta$ and
\new{$\left\| \frac{1}{|S'|}\sum_{x \in S'} (x - \mu) (x - \mu)^\top -I\right\|_{\Chi_k} \le \frac{\delta^2}{\eps}$} 
for some $0 \le \eps \le \delta$ and $\eps \le 0.5$.
Then $S'$ is $(\eps, \delta',k)$-stable with respect to $\mu$ where $\delta' = O(\delta+\sqrt{\eps})$ and $\delta' \ge \sqrt{\eps}$.
\end{lemma}

\begin{proof}
Consider an arbitrary large subset $S'' \subseteq S'$ where $|S''| \ge (1-\eps)|S'|$.
Without loss of generality, take $\mu = 0$.
Then, for an arbitrary $M \in \Chi_k$,
\begin{align*}
    \langle \overline{\Sigma}_{S''}-I, M \rangle &= \frac{1}{S''} \sum_{i \in S''} \langle x_ix_i^\top, M \rangle - 1,
\end{align*}
which is trivially at least $-1 \ge - (\delta'^2)/\eps$ for $\delta' \ge \sqrt{\eps}$.
As for the upper bound, we have
\begin{align*}
    \langle \overline{\Sigma}-I, M \rangle &= \frac{1}{S''} \sum_{i \in S''} \langle x_ix_i^\top, M \rangle - 1\\
    &\le \left(\frac{1}{S''} \sum_{i \in S'} \langle x_ix_i^\top, M \rangle\right) - 1\\
    &\le \frac{1}{1-\eps}\left(1+\frac{\delta^2}{\eps}\right) - 1\\
    &= \frac{\frac{\delta^2}{\eps}+\eps}{1-\eps}\\
    &\le \frac{2}{\eps} (\delta^2 + \eps^2)\\
    &\le \frac{\delta'^2}{\eps} \;,
\end{align*}
for some $\delta' = \Theta(\delta + \sqrt{\eps})$.

We now bound the error in the mean of $S''$ in $\ell_{2,k}$ norm.
First, observe that, for an arbitrary $k$-sparse unit vector $v$,
\begin{align*}
    \left|\frac{1}{|S'|}\sum_{i \in S'\setminus S''} \langle x_i, v \rangle \right| &= \left|\frac{1}{|S'|}\sum_{i \in S'} \1[x_i \in S' \setminus S''] \langle x_i, v \rangle \right|\\
    &\le \frac{1}{|S'|}\sum_{i \in S'} \left|\1[x_i \in S' \setminus S''] \langle x_i, v \rangle \right|\\
    &\le \sqrt{\eps} \sqrt{\frac{1}{|S'|}\sum_{i \in S'} \langle x_i, v \rangle^2}\\
    &\le \sqrt{\eps}\sqrt{1+\frac{\delta^2}{\eps}}\\
    &= \sqrt{\eps + \delta^2} \;,
\end{align*}
where the second inequality is an application of H\''{o}lder's inequality, and the third inequality uses the fact that for a unit $k$-sparse vector $v$, $vv^\top$ is in $\Chi_k$.

Thus, again for an arbitrary $k$-sparse unit vector $v$,
\begin{align*}
    \left|\langle \mu_{S''}, v \rangle\right| &=  \left|\frac{1}{|S''|} \sum_{i \in S''} \langle x_i, v \rangle\right|\\
    &\leq \frac{1}{1-\eps}\left|\frac{1}{|S'|} \sum_{i \in S''} \langle x_i, v \rangle\right|\\
    &\le 2 \left(\left|\frac{1}{|S'|} \sum_{i \in S'} \langle x_i, v \rangle\right| + \left|\frac{1}{|S'|} \sum_{i \in S'\setminus S''} \langle x_i, v \rangle\right|\right)\\
    &\le 2(\delta + \sqrt{\eps+\delta^2}) = O(\delta + \sqrt{\eps}) = \delta'.
\end{align*}
\end{proof}

\begin{proposition}[Bounded Covariance and Stability]
\label{prop:bddcovstability}
\new{Let $\mu \in \R^d$} and let $S$ be a set of $n$ samples. Let $w \in  \Delta_{n,\eps}$ over the set of samples $S$ such that \new{$\| \sum_{i} w_i (x_i - \mu)(x_i - \mu)^\top\|_{\chi_k} \leq r$} for some $r > 0$. Then $\|\mu_w - \mu\|_{2,k} \leq \sqrt{r}$.
\end{proposition}
\begin{proof}
For every $k$-sparse unit vector $v$, $vv^\top $ is in $\Chi_k$, and thus for every sparse unit vector $v$, we have that
$\sum_i w_i \langle x_i - \mu, v\rangle^2 \leq  r $. Applying Cauchy-Schwarz inequality, we get that for any sparse unit vector $v$, it follows that $\sum_i w_i \langle x_i - \mu, v\rangle \leq \sqrt{ \sum_i w_i \langle x_i - \mu, v\rangle^2} \leq \sqrt{r}$. 
\end{proof}

With \Cref{prop:bddcovstability} and \Cref{lemma:weightRounding}, we can prove \Cref{Prop:stabilityMoMCov}.

\begin{proof}[Proof of \Cref{Prop:stabilityMoMCov}]
Without loss of generality, we will assume that $\sigma=1$.
By \Cref{prop:bddcovstability}, we have that $\|\mu_w - \mu\|_{2,k} \leq \sqrt{B}$.
We thus have a weighting $w\in \Delta_{n,\eps}$, where $\|\mu_w -  \mu\|_{2,k} \leq \delta_0$
and $\|\overline{\Sigma}_w - I\|_{\Chi_k} \leq \delta_0^2/\eps$ 
for $\delta_0 = \sqrt{B}+1$, where we use triangle inequality on the $\|\cdot\|_{\Chi_k}$ norm.
By \Cref{lemma:weightRounding}, we know that there exists a set $S'$ such that $|S'| \ge (1 - 2\eps) n$ and $S'$ is $(\eps, \delta, k)$-stable with respect to $\mu$ and $\sigma$, where $\delta = O(\delta_0 + \sqrt{\eps}) = O(\sqrt{\eps} + \sqrt{B} + 1) = O(\sqrt{B} + 1)$.
\end{proof}

\subsection{Median-of-Means Pre-Processing}

This section shows \Cref{fact:MoM}, which states that the median-of-means pre-processing technique allows us to reduce to the constant-corruption case.

\label{sec:fact-mom}
\FactMoMPreProcessing*
\begin{proof}
The new algorithm simply performs median-of-means preprocessing as defined in \Cref{sec:prelim} before the fact statement, yielding $g$ new samples that are fed into the algorithm that works with constant corruption.
The uncorrupted new samples, namely the ones that are the sample mean of groups containing no originally corrupted samples, are distributed i.i.d.~according to the distribution $D'$ which has mean $\mu$, and covariance $\Sigma' = (g/n)\Sigma$, {with axis-wise fourth moment $\E_{Y \sim D'} [(Y_j - \mu_j)^4]$ being bounded by $C(g^2/n^2)\Exp_{X\sim D}[(X_j - \mu_j)^4]$ for every $j \in [d]$ for some constant $C>0$, obtained by the following fact:
	\begin{fact} (Marcinkiewicz-Zygmund inequality)
		\label{fact:MZineq}
		Recall the notation $\|X\|_{L_s}$ for a centered random variable $X$, defined  as $\E[|X|^s]^{1/s}$.
		Let $W_1,\dots,W_m, W$ be identical and
		independent centered random variables on $\R$ with a finite $\|W\|_{L_s}$ norm for $s\geq 2$. Then,
		\begin{align*}
			\left\| \frac{1}{m} \sum_{i=1}^m W_i\right\|_{L_s} \leq \frac{3 \sqrt{s}}{\sqrt{m}} \|W\|_{L_s}. 
		\end{align*}
	\end{fact}
}
First note that we give $g$ samples to the original algorithm, and $g = \Omega(\eps n) = \Omega(k^2\log d + \log(1/\tau))$ by definition.
Next, we need to check that the \emph{normalized} axis-wise 4th moment of $D'$ is {$O(1)$ times the (bound on the) $\Chi_k$-norm of the covariance matrix, that is, for all $j \in [d]$, it holds that $(\Exp_{X \sim D'}[(X_j - \mu_j)^4])^{1/4} \leq O(\sigma^4)$ and $\|\Sigma'\|_{\Chi_k} = O(\sigma^2)$.}
By the calculations at the end of the previous paragraph and the assumptions in the statement, we note {that this is true for $\sigma = O(\sqrt{g/n})$}. 

Lastly, we check that, by the scale-invariance of the original algorithm that works with constant corruption, the estimation error of the final algorithm is upper bounded {by $ O(\sigma \|\Sigma\|_{\Chi_k}) =  O(\sqrt{(g/n)\|\Sigma\|_{\Chi_k}}) = O(\sqrt{g/n}) = O(\sqrt{\eps})$ as desired.}
\end{proof}

\subsection{Basic Properties of $\Chi_k$-Norm}
\label{app:chiknorm}

The following is a straightforward bound on the $\Chi_k$-norm of a matrix based on entry-wise bounds.

\LemChiK*
\begin{proof}
Let $A = B + C $, where $B$ is a diagonal matrix and $C$ is diagonal-free.
Then we have the following using triangle inequality: $\|A\|_{\cX_k} \leq \|B\|_{\cX_k} + \|C\|_{\cX_k}$. Thus it suffices to bound each of these terms by $1$.
\begin{align*}
\|B\|_{\cX_k} \leq  \sup_{M: \sum_{i=1}^d|M_{i,i}| \leq 1} \langle B, M\rangle = \|B\|_\infty \leq \eta_1,
\end{align*}
where we use that $B$ is a diagonal matrix with entry at most $\eta_1$.
\begin{align*}
\|C\|_{\cX_k} \leq  \sup_{M: \|M\|_1 \leq k} \langle C, M\rangle = \sup_{M: \|M\|_1 \leq k} \|C\|_\infty \|M\|_1 \leq k \eta_2.
\end{align*}

\end{proof}

\section{Concentration and Truncation}
\label{app:ConcTrunc}

\subsection{Truncation Can Increase Spectral Norm of Covariance}
\label{app:truncation}

We show how truncation can increase the spectral norm of covariance from 1 to $\omega(1)$.

Consider the distribution which, with probability {$1/(2 k)$}, returns a vector where each coordinate is independent $-\sqrt{k}$ with probability $2/3$ and $2\sqrt{k}$ with probability {$1/3$}.
Otherwise, with probability $1-1/(2\sqrt{k})$, the distribution returns the origin.
The mean of the distribution is the origin, and the covariance is $I$.

Now consider the truncation $h_{0,\sqrt{k}}$, which truncates at distance $\sqrt{k}$ from the origin.
{Let $Y$ be the resulting random variable.
The mean of $Y$, $\mu'$, is thus equal to $(1/2k)(-\sqrt{k}/3,\dots,-\sqrt{k}/3) = -1/(6 \sqrt{k}) \mathbf{v}$, where $\mathbf{v}$ is the all ones vector.}
The norm of $\mu'$ is $\Theta(\sqrt{d/k})$.
Since the distribution returns the origin with constant probability (asymptotically tending to 1), the variance of $Y$ along the direction of $\mu'$, which is $\mathbf{v}/\sqrt{d}$,  is at least $\Omega(d/k)  = \omega(1)$.

\subsection{Preserving Moments under Truncation}
\label{app:TruncPreserve}

\Cref{lem:truncationInfBd} shows that truncation (mostly) preserves the mean, covariance and axis-wise fourth moments of a distribution under axis-wise fourth moment assumptions on the input distribution.

\LemTruncLInftyNorm*
\begin{proof}
Let $Y := h_{a,b}(X)$ and denote $\mu := \mu_P$. Fix an $i \in [d]$.  Since $|\mu_i - b_i| \leq a/2$ and we truncate at radius $a$, we have the following:
\begin{align}
\label{eq:truncXYreln}
    |Y_i - \mu_i| \leq |X_i - \mu_i|,\,\,\,\, \text{and}\,\,\,  |X_i - Y_i| \leq |X_i - \mu_i|.
\end{align}
Let $\cE_i$ be the event that $Y_i \neq X_i$.
We get the following by Markov's inequality and moment bounds:
\begin{align}
\label{eq:ProbE_iBd}
\cP(\cE_i) =    \P(|X_i - b_i| > a ) \leq \P(|X_i - \mu_i| \geq a/2) \leq \min\left(4\frac{\sigma^2}{a^2}, 16 \frac{\sigma^4 \nu^4}{a^4} \right) =  \min \left(\frac{\eps}{k}, \frac{\eps^2 \nu^4}{k^2}\right).
\end{align}
\begin{enumerate}
  \item 
 
We can verify the following relation using \Cref{eq:truncXYreln}:
  \begin{align}
  \label{eq:truncRel}
      |Y_i - X_i| \leq  \1_{\cE_i} \cdot \left( |Y_i - X_i| \right) \leq \1_{\cE_i} \cdot \left( |X_i - \mu_i| \right). 
  \end{align}
Applying Cauchy-Schwarz on the above inequality gives the desired conclusion:
  \begin{align*}
|\E[Y_i] - \mu_i| =  |\E[Y_i - X_i]| \leq \E \left[\1_{\cE} \cdot \left( |X_i - \mu_i| \right)\right] \leq \sqrt{\P(\cE)}\sqrt{ \E \left[ |X_i - \mu_i|^2 \right]} \leq  \sigma \sqrt{\frac{\eps}{k}},
  \end{align*}
  where we use that variance of $X_i$ is at most $\sigma^2$ and use \Cref{eq:ProbE_iBd}.
  \item This follows directly from above.
  \item By \Cref{lem:chiKSuff}, it suffices to show that $\|\Sigma_Q - \Sigma_P\|_\infty \leq 3\sigma^2 \eps \nu^4/k$.
Using triangle inequality, we obtain the following:
  \begin{align*}
      \| \Sigma_P - \Sigma_Q \|_\infty &= \Big\| \E[ (X -\mu_P)(X - \mu_P)^\top ] - \E[(Y- \mu_P)(Y - \mu_P )^\top ] \\
      &\qquad \qquad+ (\mu_Q - \mu_P)(\mu_Q - \mu_P)^\top  \Big\|_\infty \\
      &\leq \Big\| \E[ (X -\mu_P)(X - \mu_P)^\top ] - \E[(Y- \mu_P)(Y - \mu_P )^\top ]  \Big\|_\infty \\
       &\qquad \qquad+ \Big\| (\mu_Q - \mu_P)(\mu_Q - \mu_P)^\top \Big\|_\infty.
  \end{align*}
  By the first part above, we have that $\| (\mu_Q - \mu_P)(\mu_Q - \mu_P)^\top \|_\infty \leq \sigma^2 \eps/k \leq  \sigma^2\nu^4 \eps/k$, where we use that $\nu\geq 1$.
  We will thus focus on the first term. Without loss of generality, we will assume that $\mu_P = 0$ for the remainder of this proof. 
  Thus for any $i, j \in [d]$, we thus need to upper bound $\E[|X_iX_j - Y_iY_j|]$.
  \begin{align*}
    \E[|X_iX_j - Y_iY_j|] &\leq \E[|X_i\|X_j - Y_j|] + \E[|Y_j\|X_i - Y_i|] \\
    &\leq \E[|X_i\|X_j| \cdot\1_{\cE_j}] + \E[|X_i\|X_j|\cdot \1_{\cE_i}] \tag*{(Using \Cref{eq:truncRel}) } \\
    &\leq \sqrt{\E[|X_iX_j|^2]} \left( \sqrt{\P(\cE_i)} + \sqrt{\P(\cE_j)}  \right) \\
    &\leq (\E[X_i^4])^{1/4}(\E[X_j^4])^{1/4} \left(\sqrt{\P(\cE_i)} + \sqrt{\P(\cE_j)} \right)\\
    &= \sigma^2 \nu^2  \left( 2 \frac{\eps \nu^2}{k}\right)  \\
    &= \frac{2 \sigma^2 \nu^4 \eps}{k}. 
    \end{align*}
    Combining the above with \Cref{lem:chiKSuff}, we get that the $\|\Sigma_P - \Sigma_Q\|_{\Chi_k} \leq 3 \sigma^2 \eps  \nu^4$.

      \item Fix an $i \in [d]$. We use the triangle inequality and \Cref{eq:truncRel} to get the following:
      \begin{align*}
        \E[(Y - \mu_Q)_i^4] &\leq 4 (\E[(Y - \mu_P)_i^4]) +    4\|\mu_P - \mu_Q\|_\infty^4 
        \leq 4 \sigma^4 \nu^4 + 4 \sigma^4 \epsilon^2/k^2  \leq 8 \sigma^4 \nu^4, 
      \end{align*}
      where the last inequality uses that $\nu \geq 1$ and $\eps \leq 1$.
      \item This follows by definition of the random variable $Y$, the function  $h_{a,b}$, and the parameter $a$.
\end{enumerate}
 
\end{proof}

\subsection{Standard Concentration Tools}
\label{app:Conc}

\begin{fact}[VC inequality]
\label{fact:VC}
Let $\cF$ be a family of boolean functions over $\cX$ with VC dimension $r$ and let $S = \{x_1,\dots,x_n\}$ be a set of $n$ i.i.d.\ data points from a distribution $P$ over $\cX$.
If $n \gg c(r + \log(1/\tau))/\gamma^2$, then with probability $1- \tau$, for all $f \in \cF$, we have that 
\begin{align*}
\left|    \sum_{i=1}^n \frac{f(x_i)}{n} - \E_P[ f(x) ] \right| \leq \gamma.
\end{align*}
\end{fact}

\begin{lemma}[Uniform concentration over $\cA_{k,P}$]
\label{lem:unifConcAkp}
Let $S$ be a set of $n$ i.i.d.\ data points from a distribution $P$, and let $\cA_{k,P}$ be as defined in \Cref{eq:DefAkp}.
There exists a constant $c > 0$ such that if $n \geq c (k^2 \log d + \log(1/\tau))/(\AKPprob^2)$, then \Cref{eq:UnifConvAkprime} holds with probability at least $1-\tau$ over the set $S$ of $n$ i.i.d.~points from distribution $P$.
\end{lemma}
\begin{proof}
Let $Q$ be the distribution of $y:=xx^\top$. 
Let $\cF := \{ \1_{ y \cdot A > \AKPtail}: A \in \cA_{k}\}$.
Suppose for now that VC dimension of $\cF$ is less than $Ck^2\log(d)$.
Then the standard VC inequality (\Cref{fact:VC}) implies that if $n \geq c (k^2 \log d + \log(1/\tau))/(\AKPprob^2) $,
then \Cref{eq:UnifConvAkprime} holds because under $y\sim Q$, $\P (y \cdot A > \AKPtail) \leq \AKPprob $ for all $A \in \cA_{k,P}$.
Thus it remains to show an upper bound on the VC dimension of $\cF$.
Since $\cF$ corresponds to a family of linear functions that are $k^2$-sparse in $d^2$ dimensional space, \cite[Theorem 6]{AhsenVidyasagar19} implies that the VC dimension is at most $4k^2\log(3d)$.
This completes the proof.
\end{proof}

\section{Choice of Numerical Constants}
\label{app:consts}

This section shows how to pick the numerical constants $\AKPprob,\AKPtail,\Qconst,\Qscale,V_Z$ and $\res$.
In the proof of \Cref{thm:stabHighProb}, these constants need to satisfy the following constraints:
\begin{enumerate}
    \item $\Qscale \geq 2$.
    \item $\AKPprob $ is at least a small constant since  the sample complexity is inversely proportional to $1/\AKPprob^2$.
    \item See \eqref{eq:QinAk}: \begin{align*}
            \frac{1}{\Qconst} + \frac{\Qconst}{\Qscale^2} \leq 10^{-6}.
    \end{align*}
    \item See \eqref{eq:AkpProbCondition}: \begin{align*}
\frac{\sigma^2}{\AKPtail} + \frac{4}{\Qscale} + \frac{\Qconst \times \nu^4}{ \Qscale \times \AKPtail^2} \leq 10^{-6} \times \AKPprob.
    \end{align*}
    
    \item See \eqref{eq:resVsqscaletail}: $\res \geq \Qscale \times \AKPtail + 10 \sqrt{V_Z}$.
    
    \item See \eqref{eq:BerryEsseen}: $\frac{0.57 (\sigma^2 \times \InftyBd^2 \times \Qconst)}{\sqrt{V_Z}} \le 0.1$.
    
    \item Paley-Zygmund:  $0.004 \geq 4 \times \AKPprob$.
\end{enumerate}

\noindent Therefore, we pick the constants as follows:
\begin{enumerate}
    \item $\nu,\sigma$ and $\InftyBd$ are the numbers we get from the $\ell_\infty$ truncation, and thus there is nothing to choose here.
    \item $\AKPprob = 0.001$.
    \item $\Qconst = 10^7$.
    \item $\Qscale = 10^{10}$.
    \item Solve for $\AKPtail$ in terms of above in Constraint 4. It suffices to take $\AKPtail = \max(\sigma^2,\nu^2)\times 10^{10}$.
    
    \item Solve for $\sqrt{V_Z}$ using Constraint 6. It suffices to take $V_Z = 10^{16} \sigma^4 r^4$.
    \item Solve for $\res$ using Constraint 5. It suffices to take $\res = \max(\sigma^2,\sigma^2 r^2,\nu^2)\times 10^{20}$.
\end{enumerate}

\end{document}